\documentclass[11pt]{article}

\usepackage[english]{babel}
\usepackage[letterpaper, margin=1in]{geometry}
\usepackage[colorlinks=true, allcolors=blue]{hyperref}

\usepackage{graphicx}
\usepackage{amsmath}
\usepackage{amsthm}
\usepackage{amssymb}
\usepackage{booktabs}
\usepackage{booktabs}
\usepackage{algorithm}
\usepackage{algorithmic}
\usepackage{dirtytalk}
\usepackage{thmtools}
\usepackage{thm-restate}
\usepackage{subcaption}
\usepackage{caption}
\usepackage{enumerate}
\usepackage{dirtytalk}
\usepackage[most]{tcolorbox}

\usepackage{pgf}
\usepackage{lmodern}
\makeatletter

\everymath=\expandafter{\the\everymath\displaystyle}
\makeatother

\newtheorem{theorem}{Theorem}
\newtheorem{lemma}{Lemma}
\newtheorem{definition}{Definition}

\newtheorem{remark}{Remark}
\newtheorem{corollary}{Corollary}
\newtheorem{claim}{Claim}

\newcommand{\s}{\textsf{\bf s}}
\newcommand{\f}{f}
\newcommand{\x}{\textsf{\bf x}}
\newcommand{\spot}{\textsf{s}^\diamond}
\newcommand{\poa}{{\sf PoA}}
\newcommand{\spred}{\bar{\s}}
\newcommand{\fpred}{\bar{\f}}
\newcommand{\fpot}{\f^\diamond}

\allowdisplaybreaks[4]

\title{Learning-Augmented Coordination Mechanisms}
\date{}

\author{
\textbf{George Christodoulou}$^{1,3}$
\qquad
\textbf{Vasilis Christoforidis}$^{1,3}$
\\[0.35em]
\textbf{Alkmini Sgouritsa}$^{2,3}$
\qquad
\textbf{Ioannis Vlachos}$^{2,3}$
\\[0.9em]
\small $^{1}$Aristotle University of Thessaloniki, Thessaloniki, Greece
\\
\small $^{2}$Athens University of Economics and Business, Athens, Greece
\\
\small $^{3}$Archimedes, Athena Research Center, Athens, Greece
\\[0.6em]
\small
\texttt{\{gichristo,vgchrist\}@csd.auth.gr}
\qquad
\texttt{\{alkmini,ioa.vlahos\}@aueb.gr}
}

\begin{document}

\maketitle

\begin{abstract}
    The inefficiency of decentralized resource allocation, a core challenge in Algorithmic Game Theory, is classically measured by the {\em Price of Anarchy (PoA)} and the {\em Price of Stability (PoS)}. The most widely studied and foundational class of these models are {\em atomic and non-atomic congestion games}, for which a robust theory quantifying the inefficiency of equilibria is well-established.
To mitigate the effects of this selfish behavior, we employ {\em Coordination Mechanisms}, which modify resource costs to incentivize
socially improved equilibrium outcomes. In their standard form,
Coordination Mechanisms have been limited by a pessimistic, worst-case
view; it is assumed that the demand is {\em unknown} and {\em adversarially}
chosen. Consequently, positive results remain sparse, applying
mainly to specialized network topologies like parallel links.
We address this limitation by studying the design and analysis of {\em learning-augmented coordination mechanisms}, where the mechanism is
endowed with a potentially inaccurate {\em prediction of the demand
  $\bar{r}$}. We study general {\em atomic and non-atomic congestion
  games} with polynomial latency functions of degree $d$. Our main
positive result is a learning-augmented coordination mechanism tuned
by a confidence parameter $\beta$ which achieves a "best of both worlds" type of result: it
    achieves approximately optimal performance for an accurate
    prediction without sacrificing the worst-case guarantee of a good
    equilibrium (PoS). We further establish bounds on a stronger
    robustness measure, the PoA-Robustness (worst equilibrium
    performance under cost modification). The mechanism's construction
    relies on transforming approximate equilibria of the original game
    into exact equilibria of the modified game using a carefully
    defined {\em approximate potential function}.
We complement our positive results by providing {\em impossibility
  results}. We first show that a stronger notion of
consistency, requiring that {\em all} equilibria have optimal
social cost for an accurate prediction, fails to guarantee bounded
robustness. We also show a strong {\em error tolerance} limit: any
mechanism achieving optimal consistency in non-atomic games suffers a
sharp, discontinuous performance that collapses to the original {\em PoA} value under even infinitesimally small prediction errors,
highlighting the limits of efficiency guarantees in non-atomic
congestion games.

\end{abstract}

\newpage

\tableofcontents

\newpage

\section{Introduction}
For the past three decades, a core research area in Algorithmic Game
Theory has focused on investigating and quantifying the \emph{inefficiency
of equilibria}. Two standard concepts are used to measure this
inefficiency. On the one hand, the \emph{Price of Anarchy} (PoA)
\cite{KoutsoupiasP99} serves as a robust measure for this assessment,
comparing the worst-case equilibrium with the system optimal
solution. On the other hand, the \emph{Price of Stability}
(PoS) \cite{siamcomp/AnshelevichDKTWR08:PoS} offers an optimistic
measure that uses the best-case equilibrium for this comparison. The
former is a very robust measure, while the latter is perhaps a more appropriate
concept from an engineering perspective, analyzing the ideal solution
that we would like our protocols to produce.

The richest and most well-studied classes of games among these
models are arguably \emph{atomic and non-atomic congestion games} (see, e.g.,
\cite{NRTV2007a} for a detailed discussion), where a set of players
compete for a set of resources. The resources are sensitive to the
amount of traffic they receive, which is typically modeled via
a latency function. This highly descriptive class of games captures a
wide variety of application domains, most prominently \emph{routing games}.

The PoA and PoS of congestion games are thoroughly
understood. Roughgarden and Tardos's seminal work
\cite{RoughgardenT02, RT04} delivered tight bounds for the PoA of \emph{non-atomic
congestion games}, where individual players exert negligible influence
on traffic. For \emph{atomic congestion games}, the PoA was progressively
resolved: initially for affine latency functions
\cite{stoc/AwerbuchAE05, stoc/ChristodoulouK05}, then for polynomials
\cite{siamcomp/AlandDGMS11}, and finally for more general monotone
cost functions \cite{Roughgarden15}. These developments underscore
congestion games' role as a central paradigm within the
\emph{Smoothness Framework} \cite{Roughgarden15}, an elegant and unified
methodology for bounding the PoA.
The PoS\footnote{We note that for non-atomic congestion games
  PoA and PoS coincide.} of congestion games retains much smaller
values compared to the respective PoA values: approximately $1.577$
for affine functions \cite{esa/ChristodoulouK05,algorithmica/CFKKM11},
as opposed to the PoA bound of $2.5$; and $\Theta(d)$
\cite{teac/CG16:PoS} for polynomials of degree $d$ as opposed to the
much higher PoA bound of $\left(\Theta\!\left(d/\ln d\right)\right)^{d+1}$. All aforementioned PoA
and PoS results are {\em tight}, and collectively they provide a solid theory
for estimating the inefficiency of equilibria in congestion
games.

\textbf{Coordination Mechanisms.}
While the PoA and PoS offer valuable
insights into the performance of equilibria, they are inherently
\emph{post-factum} measures; they assess the quality of an equilibrium
that has already been reached. This observation motivated the
introduction of \emph{coordination mechanisms}
\cite{tcs/ChristodoulouKN09:CMs, algorithmica/ChristodoulouMP14} for
congestion games. A coordination mechanism aims to improve the equilibrium performance
of a given congestion game by modifying (e.g., increasing) the latency
functions of the resources in a way that encourages equilibrium
behavior with good social cost. The challenge is that these
modifications must occur \emph{without prior knowledge of the traffic
demand}.\footnote{In a routing game, the demand refers to the amount of
traffic that needs to move from each source to each destination of a
given network. If the demand is known in advance,
it is possible to modify the latencies to enforce the optimal
flow~\cite{algorithmica/ChristodoulouMP14}.} \cite{algorithmica/ChristodoulouMP14}
took the worst-case approach where demand is chosen adversarially after
latency modifications are applied. This robust, yet overly pessimistic, approach has yielded positive results only in the special case of
parallel links with affine
latencies~\cite{algorithmica/ChristodoulouMP14}. However, in many
practical settings, it is often realistic to assume that traffic can
be accurately predicted from historical data \cite{LvDKLW15:trafficPrediction}. 

In light of this, we adopt the \emph{learning-augmented framework} (see,
e.g., \cite{cacm/MitzenmacherV22} for an introduction), where
algorithms are equipped with an additional, potentially unreliable,
prediction about the input. This paradigm has proven effective in the
design of \emph{online algorithms}, where the algorithm is augmented with a
prediction of future inputs. In this work, we study the design of
\emph{learning-augmented coordination mechanisms}, where the mechanism is augmented with a \emph{prediction of the
demand rate}.\footnote{The work of \cite{CCSV25:predictions} is closely related to our work and we compare our results with the results of \cite{CCSV25:predictions} in Section~\ref{sec:related}.}

\subsection{Contributions}

We study coordination mechanisms through the lens of the learning-augmented
framework. A learning-augmented coordination mechanism takes as input
an instance of a congestion game and a {\em prediction of the demand $\bar{r}$} (see Sec.~\ref{sec:prelims} for formal definitions), but
not the actual demand $r$ which is unknown.
Our results apply to a broad class of games: non-atomic (Section
\ref{sec:nonAtomic}) and atomic (Section \ref{sec:atomic}) congestion
games with polynomial latency functions with nonnegative coefficients
of maximum degree $d$. 

\subsubsection{Positive Results}
\textbf{Consistency - Robustness trade-off.} We propose a
learning-augmented coordination mechanism, called
\textsc{Potential}, which takes as input a predicted demand $\bar{r}$ and can be tuned by
a {\em confidence parameter} $\beta\in [1,d+1]$. \textsc{Potential}
modifies (increases) the latency functions of the game to ensure that the modified game satisfies the following:
\begin{itemize}
\item    when the prediction is accurate (i.e., $r=\bar{r}$),  there exists an equilibrium with approximately optimal social cost as a function of $\beta$ {\bf (Consistency)};

\item even when the predicted
  demand is {\em arbitrarily inaccurate}, there exists an equilibrium whose social
  cost is at most a factor of $\beta(d+1)$ away from the optimal social cost 
  {\bf (Robustness)}. 
  \end{itemize}

  We remark that in both cases, we compare the equilibrium cost with
  the {\em increased} latencies to the {\em original} optimal cost, prior
  to the modification.  The exact statements can be found in Theorem~\ref{thm:consrobust} for the case of non-atomic congestion
  games and in Theorem~\ref{thm:atomic} for the case of atomic
  congestion games, respectively.

It is instructive to compare our results with the performance of the best Nash Equilibrium in the original game which
is measured by the PoS. For atomic congestion games, we are able to improve the
performance of the best equilibrium by a factor of $\beta$ when the
prediction is accurate, while for arbitrarily inaccurate predictions, we recover the worst-case guarantee for the best equilibrium up to a factor of $\beta$. For non-atomic games we obtain analogous guarantees, by losing at most a factor depending on $\beta$ and a small factor of $O(\ln(d))$. In Table~\ref{tbl:results1} we provide the asymptotic PoS upper bounds in the original games as $d\rightarrow\infty$, alongside upper bounds on consistency and robustness in the modified games. For the non-atomic consistency entry, the table reports a simplified upper bound;
the exact expression is given in Theorem~\ref{thm:consrobust}.

\begin{table}[ht]
    \centering
    \begin{tabular}{lccc}
        \toprule
        Model 
        & PoS ($d\rightarrow\infty$) 
        & Consistency 
        & Robustness \\
        \midrule
        Non-atomic 
        & $\Theta\!\left(\frac{d}{\ln d}\right)$ \cite{RT04}
        & $\frac{d+1}{\beta}$ (Thm.~\ref{thm:consrobust})
        & $\beta(d+1)$ (Thm.~\ref{thm:consrobust}) \\[4pt]
        Atomic 
        & $\Theta(d)$ \cite{teac/CG16:PoS}
        & $\frac{d+1}{\beta}$ (Thm.~\ref{thm:atomic})
        & $\beta(d+1)$ (Thm.~\ref{thm:atomic}) \\
        \bottomrule
    \end{tabular}
    \caption{Summary of the main guarantees for atomic and non-atomic congestion games with polynomial latency functions of degree $d$. The first column reports the classical PoS benchmark without predictions, while the last two columns give the consistency and robustness guarantees of our learning-augmented coordination mechanism with confidence parameter $\beta$. For non-atomic games, as $d \to \infty$, the consistency bound can be refined to $O\!\left(d/(\beta+\ln d)\right)$.}
    \label{tbl:results1}
\end{table}

Our results demonstrate that incorporating predictions of the input may dramatically improve the performance when the prediction is accurate, without significantly compromising robustness for moderate values of $\beta$. We further emphasize that our results hold for {\em general} congestion games with polynomial latency functions, as opposed to previous work on coordination mechanisms which considered parallel links \cite{algorithmica/ChristodoulouMP14,isaac/WangYC16:ImprovedTax,CCSV25:predictions} and affine latency functions \cite{algorithmica/ChristodoulouMP14}.

\textbf{PoA-Robustness.}
We also provide a bound on the performance guarantee of {\em any}
equilibrium under arbitrarily inaccurate predicted demand ({\bf PoA-Robustness}). As before, we compare the worst-case equilibrium cost measured using the {\em modified} latencies, 
with the {\em original} optimal cost, prior
  to the modification. In Table~\ref{tbl:results2} we provide the asymptotic PoA bounds for both the unmodified and modified games when $d\rightarrow\infty$.

\begin{table}[ht]
    \centering
    \begin{tabular}{lcc}
        \toprule
        & PoA & PoA-Robustness  \\
        \midrule
        Non-atomic 
        & $\Theta\!\left(\frac{d}{\ln d}\right)$ \cite{RT04}
        & $\max\{d+1,\beta^{d+1}\}$
        (Thm.~\ref{thm:consrobustPoA}) \\[4pt]
        Atomic 
        & $\left(\Theta\!\left(\frac{d}{\ln d}\right)\right)^{d+1}$
        \cite{siamcomp/AlandDGMS11}
        & $O\!\left(\left(\frac{d}{W(d/\beta)}\right)^{d+1}\right)$
        (Thm.~\ref{thm:atomicPoA})  \\
        \bottomrule
    \end{tabular}
    \caption{Asymptotic behavior, as $d\to\infty$, of the classical PoA bounds for atomic and non-atomic congestion games with polynomial latency functions of degree $d$, together with the corresponding PoA-Robustness guarantees of our learning-augmented coordination mechanism. For non-atomic games, the displayed bound is the simplified uniform upper bound implied by Theorem~\ref{thm:consrobustPoA}: the exact first branch is $\Theta(d/\ln d)$ at $\beta=1$ and increases up to $d+1$ as $\beta$ approaches $(d+1)^{1/(d+1)}$, while the second branch is $\beta^{d+1}$. In the atomic PoA-Robustness bound, $W$ denotes the Lambert $W$ function.}
    \label{tbl:results2}
\end{table}

Note that, for atomic games, the performance degrades in the worst-case only from $\left(\Theta\!\left(d/\ln d\right)\right)^{d+1}$
to $\left(O(d)\right)^{d+1}$. For non-atomic games, the PoA-Robustness can be substantially worse than in the unmodified game. This degradation occurs because our intervention introduces discontinuities, effectively transforming the game into an atomic one. Note, however, that such discontinuities are necessary: without them, no improvement is possible, as established in \cite{algorithmica/ChristodoulouMP14}. Nonetheless, we emphasize that all our bounds are essentially constant, as they do not depend on the demand (i.e., the number of agents or the input rate) but only on the maximum degree $d$ of the polynomial.

\subsubsection{Impossibility Results}

We complement our results by providing impossibility results that hold for {\em all} learning-augmented coordination mechanisms.

\textbf{PoA - Consistency.}
At this point, a natural question may arise: \say{{\em Is it possible
    to obtain good consistency - robustness trade-off if one insists
    that {\em all} equilibria of the modified game have optimal social
    cost when $r=\bar{r}$?}} We present strong negative results
for both atomic (Theorem~\ref{thm:impossibilityPoAcons}) and non-atomic games (Theorem~\ref{thm:impossibilityPoAconsNon-atomic}), showing that any coordination mechanism that
demands this stronger notion of consistency fails to guarantee 
bounded robustness for any arbitrarily inaccurate predicted demand.

\textbf{Error tolerance.} 
Lastly, inspired by the notion of {\em smoothness} for learning-augmented algorithms, we study how approximation guarantees degrade under small prediction errors in non-atomic congestion games (see Section \ref{sec:smoothness}). 
In Theorem \ref{thm:smoothNonSymmetric}, we establish a strong impossibility result, showing that there is {\em no} mechanism
that can achieve optimal performance under perfectly accurate predictions while simultaneously maintaining approximately optimal performance when the prediction is slightly inaccurate. More precisely, for any $1$-consistent mechanism, we show that even an infinitesimal prediction error causes the performance to revert back to the original PoA obtained without predictions. Furthermore, we show that this negative result applies to any class of games beyond single-commodity network congestion games (see Section \ref{sec:smoothness}). For instance, in multi-commodity networks or symmetric congestion games, no mechanism can be $1$-consistent while also maintaining an approximation ratio that approaches $1$ as the prediction error tends to zero. Nonetheless, in the special case of single-commodity network congestion games, we exploit the network structure to obtain a positive result (Theorem~\ref{thm:SmoothnessSC}).

\subsubsection{Techniques}

\textbf{Design of coordination mechanisms.}
At a high level, our coordination mechanism transforms {\em
  approximate} equilibria with approximately optimal social cost in
the original game into {\em exact} equilibria of the modified game that
retain the same social cost.
A central building block in the definition of our mechanism is an
{\em approximate potential function} w.r.t. a twisted set of latency
functions $h$, tuned by the confidence parameter $\beta$. The
minimizers of this approximate potential are approximate equilibria
with bounded performance. The associated latencies $h$, together with the
potential minimizer $\spred$ (for $r=\bar{r}$), are carefully
used via a threshold mechanism to define the \textsc{Potential}
 coordination mechanism. The construction ensures that, when $r=\bar{r}$, the  potential minimizer is also a potential minimizer of the modified game, and hence an exact equilibrium for that game, while retaining its original social cost. 
Remarkably, this improvement comes at little cost: even under inaccurate predictions, any potential minimizer of the modified game continues to provide strong performance guarantees.

\textbf{Analysis.}
When predictions are inaccurate, we propose a framework to provide
robustness bounds for learning-augmented coordination mechanisms. To
achieve this, we transfer standard techniques for deriving bounds on the inefficiency of equilibria, most notably the smoothness framework for the PoA and the potential function method for the PoS, to relate the equilibria of the modified game to those of the original game.

We remark that the games that we consider in this work are guaranteed to possess pure Nash equilibria. For simplicity, our upper bound proofs are presented
 for pure equilibria. However, all our proofs can be directly
extended to mixed, correlated, and coarse-correlated equilibria via standard
extension arguments. 

Our approach is quite general and provides a unified framework for
designing and analyzing coordination mechanisms in both non-atomic and
atomic congestion games. In most cases, we tailor the analysis to the
specific parameters of an approximate potential function, yielding
tighter bounds than those obtained through general approximate
equilibrium guarantees.

\subsection{Further Related Work}
\label{sec:related}
\paragraph{Inefficiency of equilibria.}
Of particular relevance to our work is the study of the performance degradation resulting from selfish behavior and methods to mitigate it. 
Extensive work has characterized the PoA and PoS of congestion games; apart from the works already discussed in the introduction, bounds were shown for mixed and (coarse) correlated equilibria \cite{esa/ChristodoulouK05,stoc/BlumHLR08,Roughgarden15}, as well as for approximate equilibria \cite{arxiv:CKS08} and weighted congestion games \cite{teac/BhawalkarGR14:PoAWeighted,siamcomp/ChristodoulouGG19,tocs/GiannakopoulosP23}. 

\paragraph{(Approximate) Potentials.} Since the work of
\cite{Rosenthal1973}, that introduced the class of congestion games, a long line of work has been devoted to identifying classes of games that admit pure Nash equilibria. A key
technical tool of Rosenthal's work was the use of a \emph{potential}
function which captures the gains in each player's deviating moves in
an elegant way. Congestion games are \emph{exact potential} games
\cite{MondererShapley:Potential}. Beyond proving the existence of
equilibria, potential functions can also provide important structural
properties on the performance of potential minimizers.

This has motivated the development of \emph{approximate} potential functions, which extend the power of potential-based arguments to approximate equilibria (see e.g. \cite{algorithmica/ChristodoulouKS11,approx/HansknechtKS14,siamcomp/ChristodoulouGG19,CaragiannisF21,mor/GiannakopoulosNS22,tocs/GiannakopoulosP23}). Moreover, the well-known fact that \emph{weighted} congestion games do not always have equilibria in pure strategies \cite{tcs/FotakisKS05} further underlined the need for approximate potentials. Minimizers of approximate potential functions are guaranteed to be approximate equilibria, in contrast to Rosenthal's case.

\paragraph{Coordination mechanisms.}\label{par:coordination_mechs}

Coordination mechanisms were first introduced in \cite{tcs/ChristodoulouKN09:CMs} as a method to mitigate the inefficiency of equilibria. They have since been extensively studied in the design of local policies across various scheduling settings and congestion games \cite{tcs/ImmorlicaLMS09,mst/CohenDT11,esa/AbedH12,algorithmica/Caragiannis13,Kollias13,esa/AbedCH14,BhattacharyaIKM14,icalp/BhattacharyaKM14,ColeCGMO15,AzarFJMS15,tocs/CaragiannisF19,soda/Kashaev26}. Specifically, \cite{algorithmica/ChristodoulouMP14,isaac/WangYC16:ImprovedTax} considered coordination mechanisms in the context of \emph{selfish routing} over parallel-link networks.
\cite{CaragiannisGV17} established intricate connections between cost-sharing protocols and coordination mechanisms, while \cite{ior/CGS24} applied the main principle of coordination mechanisms to cost-sharing games, aiming to improve social welfare by overcharging.

\paragraph{Coping with Selfishness.} Besides coordination mechanisms,
an important approach explored in the literature is the use of \emph{taxation} and \emph{tolls}, where an agent may be charged a tax that influences their perceived latency. The idea of introducing taxes to improve system efficiency has been extensively studied in transportation science (cf. \cite{stoc/ColeDR03} and references therein), and gained significant traction in the algorithmic game theory literature; both in non-atomic games (see e.g., \cite{focs/FleischerJM04,focs/KarakostasK04,jcss/ColeDR06Taxes}) and atomic ones (e.g., \cite{talg/CKK10,FotakisSpirakis08,teac/BV19,ior/PG24}). \cite{icalp/Colini-Baldeschi18} deviate from the marginal edge pricing approach and seek for optimal tolls without the exact knowledge of the demand. They show conditions under which demand-independent optimal tolls exist for any input rate in the network. Another approach suggests that the game designer can select a fraction of players to be centrally routed using a \emph{Stackelberg strategy} \cite{siamcomp/Roughgarden04,mor/BonifaciHS10,tocs/Fotakis10a}. \emph{Signaling} strategies have also been used to leverage the designer's informational advantage and steer the traffic equilibrium to a more efficient outcome \cite{sigecom/BhaskarCKS16}. \emph{Cost-sharing protocols} determine how a common resource is allocated among players \cite{or/GkatzelisKR16}.

\paragraph{Learning-augmented framework.} The learning-augmented framework has gained significant attention in the design of algorithms beyond the worst-case. Rather than surveying the growing literature, we highlight a few representative works on canonical problems such as auctions, strategic facility location, and strategic scheduling \cite{mor/AgrawalBGOT24,ITCS/BalkanskiGT23,CBKST24,CSV24,sigecom/LuW024,soda/Gkatzelis0T25}. 
In the direction of improving the PoA, \cite{sigecom/GkatzelisKST22} incorporated demand predictions into the design of cost-sharing mechanisms for scheduling games. For a broader overview, we refer the interested reader to \cite{cacm/MitzenmacherV22,sigecom/BalkanskiGT23Overview} and the references therein, and to \cite{LM:alpsRepo} for a frequently updated online repository of works in this area.  

The recent paper~\cite{CCSV25:predictions} is closely related to our
work. They also explore the use of demand predictions in the design of
coordination mechanisms for games with polynomial latencies, but their
results are incomparable with ours. The focus
of~\cite{CCSV25:predictions} is on {\em non-atomic} congestion games
and on {\em parallel-link} networks, while our results apply to a much
more general setting; to {\em general} congestion games, for both
atomic and non-atomic variants. However, they are able to obtain
stronger results for the stronger definitions of consistency and
robustness. They are able to show an improved PoA-Robustness upper bound of
$d+1$ for $1$-PoA-Consistent mechanisms and they also show that this
result is tight.  As we show in Theorems~\ref{thm:impossibilityPoAcons} and \ref{thm:impossibilityPoAconsNon-atomic} (for atomic and non-atomic congestion games, respectively), it is impossible to achieve this in general congestion games.

\section{Model and Notation}
\label{sec:prelims}
We study both atomic and non-atomic congestion games. We introduce the
relevant notation and the definitions of each setting separately.
\paragraph{Atomic congestion games.}
An (atomic) {\em congestion game} is a tuple
$(N,E, (\mathcal{S}_i)_{i\in N},(\ell_e)_{e\in E} )$, where
$N = \{1, \dots, n\}$ is a set of $n$ players and $E$ is a set of
resources. Each player $i\in N$ chooses a pure {\em strategy}
$s_i\subseteq E$, from a set of available strategies
$\mathcal{S}_i \subseteq 2^E$. Each resource $e\in E$ is associated with a
nonnegative {\em latency function}
$\ell_e : \mathbb{N}\to \mathbb{R}^+$.  We consider polynomial latency
functions with maximum degree $d$ and nonnegative coefficients; that
is every latency function is of the form
$\ell_e(x)=\sum_{k=0}^d a_{e,k} \cdot x^k$ with $a_{e,k}\ge 0$ for all
$k$. We refer to the pair $r=(N,(\mathcal{S}_i)_{i\in N})$ as the {\em demand} of the game.

A \emph{pure strategy profile} is a choice of strategies
$\s = (s_1, s_2, \dots, s_n)\in \mathcal{S}=\mathcal{S}_1\times \cdots \times \mathcal{S}_n$
one for each player. We adopt the standard notation
$\s_{-i}=(s_1,\ldots, s_{i-1},s_{i+1}, \ldots s_n)$,
$\mathcal{S}_{-i}=\mathcal{S}_1\times \cdots \times \mathcal{S}_{i-1} \times
\mathcal{S}_{i+1} \times \cdots \times \mathcal{S}_n$, and $\s=(s_i,\s_{-i})$.
For a pure strategy profile $\s$ define the \emph{load}
$x_e(\s)=|i\in N: e\in s_i|$ as the number of players that use
resource $e$. The \emph{cost} for player $i$ is defined by
$C_i(\s)=\sum_{e\in s_i} \ell_e(x_e(\s))$.

A pure strategy profile $\s$ is a pure \emph{Nash equilibrium} if and only if for every
player $i\in N$ and for all $s'_i\in \mathcal{S}_i$, we have $C_i(\s)\leq C_i(s'_i,\s_{-i})$.
The social cost of a strategy profile $\s$ is the sum of the players' costs 
$$C(\s) = \sum_{i \in N}C_i(\s)=\sum_{e \in E}x_e(\s)\ell_e(x_e(\s)).$$ We
further denote by $C_{opt}$ the minimum social cost,
i.e. $C_{opt}=\min_sC(s)$.
The \emph{Price of Anarchy} of a congestion game is the social cost of the worst-case Nash equilibrium over the optimum social cost
$$ \poa = \max_{\s \in \text{Eq}(r)}\frac{C(\s)}{C_{opt}},$$ where Eq$(r)$ denotes the set of all equilibria under demand $r$.
The PoA for a class of games is the largest PoA among all games in the class.

Congestion games admit a {\em potential} function
$\Phi(\s)=\sum_{e\in E}\Phi_e(\s)=\sum_{e\in
  E}\sum_{j=1}^{x_e(\s)}\ell_e(j),$ which was introduced by
Rosenthal~\cite{Rosenthal1973}. The set of pure Nash equilibria corresponds to local optima
of the potential function~\cite{Rosenthal1973}.  This mapping suggests a special pure Nash equilibrium,  the {\em potential minimizer}, which is the profile
$\s^\diamond$ that minimizes the potential, i.e., $ \Phi(\spot) \le \Phi(\s') $ for all $ \s'\in \mathcal{S}$.
We sometimes use the notation
$\Phi^\ell, \Phi_{e}^{\ell}$ to emphasize the dependency of the latency function in the computation. Similarly, non-atomic congestion games admit a potential function $\Phi(f) = \sum_{e \in E}\int_0^{f_e}\ell_e(x)dx$ and equilibria correspond to global potential minimizers of the potential function~\cite{Beckmann1956-}.
The inefficiency of potential minimizers has been well-analyzed as it is used as a means to show upper bounds for the Price of Stability\footnote{The price of stability of a game is the social cost of the best-case Nash equilibrium over the optimum social cost.}~\cite{siamcomp/AnshelevichDKTWR08:PoS}, see e.g. \cite{esa/ChristodoulouK05,algorithmica/ChristodoulouKS11,algorithmica/CFKKM11,teac/CG16:PoS,Bilo18,siamcomp/ChristodoulouGG19}.

\paragraph{Non-atomic congestion games.} 
Rather than formally defining the class of non-atomic congestion
games, we focus on the selfish routing problem, as it serves as a
clear and more illustrative example. The key difference is that, in
non-atomic congestion games, players' strategies are not necessarily
paths within a network (as in the selfish routing problem), but rather
arbitrary subsets of resources.%

Let $G=(V,E)$ be a directed graph, where $V$ is the set of vertices
and $E$ is the set of edges. We consider a set $\mathcal{I}$ of $k$
commodities: origin-destination pairs $\{s_i, t_i\}_{i\in \mathcal{I}}$
that define the origins and the destinations in the graph. We denote
by $\mathcal{P}_i$ the set of all $s_i-t_i$ simple paths of $G$, and by
$\mathcal{P} = \bigcup_i \mathcal{P}_i$ the union of all feasible
paths. A flow $f$, is a mapping from the set of paths to the set of nonnegative
reals $f:\mathcal{P}\rightarrow\mathbb{R}_{\ge0}$. For a given flow $f$,
the flow on an edge is defined as the sum of the flows of all the
paths that use this edge $f_e=\sum_{P\in \mathcal{P}, e\in P}f_P$. We
relate with every commodity $\{s_i, t_i\}$ a demand rate $r_i$, as the
total traffic that needs to move from $s_i$ to $t_i$. A flow $f$ is
feasible, if for every commodity $\{s_i, t_i\}$, the traffic rate
equals the total flow routed on paths in $\mathcal{P}_i$, i.e., $r_i=\sum_{P\in
  \mathcal{P}_i}{f_P}$. We refer to $r=(r_i)_{i\in \mathcal{I}}$ as the {\em demand} of the game.

Every edge introduces a delay in the network which depends on the load
of the edge and is determined by a nonnegative latency function,
$\ell_e : \mathbb{R}_{\ge0}\to \mathbb{R}_{\ge0}$. In our case this is a polynomial of maximum degree
$d$ with nonnegative coefficients. An instance of a routing game is
denoted by the triple $(G,r,\ell)$. The latency of a path $P$, for a
given flow $f$, is defined as the sum of all the latencies of the
edges that belong to $P$, i.e., $\ell_P(f)=\sum_{e\in P}\ell_e(f_e)$. The
social cost w.r.t. a given flow $f$ is the total delay due to
$f$
\[C(f)=\sum_{P\in \mathcal{P}}\ell_P(f)f_P=\sum_{e\in E}\ell_e(f_e)f_e.\]
We denote by $C_{opt}(r)$ the minimum cost of a flow for demand $r$.
A feasible flow $f$ that routes $r_i$ units of flow from $s_i$ to $t_i$ for each commodity $i\in \cal{I}$,
is at Nash (or Wardrop) equilibrium if and only if it is routed on
shortest $s_i-t_i$ paths, i.e., if for any two paths
$p_1, p_2 \in \mathcal{P}_i$ with $f_{p_1} > 0$,
$\ell_{p_1}(f) \le \ell_{p_2}(f)$.
We will often use discontinuous latency functions, under which
standard Wardrop equilibria might not exist. Hence, we resort to the
notion of {\em user
  equilibria}~\cite{transci/BernsteinS94:UserEq}. A feasible flow $f$ that routes $r_i$ units of flow for each
commodity $i \in \mathcal{I}$ is a user equilibrium, if and only if
for all paths $p_1, p_2 \in \mathcal{P}_i$ with $f_{p_1}>0$,
\[\ell_{p_1}(f) \le \liminf_{\epsilon \downarrow 0}\ell_{p_2}(f + \epsilon\boldsymbol{1}_{p_2} -\epsilon\boldsymbol{1}_{p_1}), \]
where $\boldsymbol{1}_p$ denotes a unit of flow passing through $p$. Throughout the paper, we refer to user equilibrium simply as equilibrium.

\paragraph{Learning-augmented coordination mechanisms.}
A {\em learning-augmented coordination mechanism} takes as input a
partial instance of a game, where the actual demand $r$ is missing. This instance is augmented with a \emph{prediction}\footnote{
This is a key difference with the standard definition of a coordination mechanism~\cite{tcs/ChristodoulouKN09:CMs, algorithmica/ChristodoulouMP14}, where $r$ is an unknown parameter chosen adversarially.} of the demand $\bar{r}$. For the
case of \emph{atomic games}, this input is a tuple
$(\bar{N},E, (\bar{\mathcal{S}}_i)_{i\in \bar{N}},(\ell_e)_{e\in E} )$, where the predicted demand $\bar{r}$ consists of a prediction of the set of players $\bar{N}$, and a
prediction of the strategy space\footnote{In a network congestion game, the strategy space can be implicitly defined by the simple paths of the network that connect a set of origin-destination pairs $(\bar{s}_i,\bar{t}_i)$, one for each $i\in \bar{N}$.} $\bar{\mathcal{S}}_i$ for each player $i\in \bar{N}$. For the case of
\emph{non-atomic games}, the input is a tuple $(G,\bar{r},\ell)$ where $\bar{r} = (\bar{r}_i)_{i\in \bar{\mathcal{I}}}$ refers to the prediction on the set of commodities
$\bar{\mathcal{I}}$ along with a prediction of the demand $\bar{r}_i$
for each commodity $i\in \bar{\mathcal{I}}$. For simplicity, we will often refer only to the predicted demand $\bar r$ as the input of the coordination mechanism, emphasizing the predicted element of the input, as the other parameters are usually inferred from the context.

Given a predicted demand $\bar{r}$, a learning-augmented mechanism defines a
set of modified latencies $\hat{\ell}=(\hat{\ell}_e)_{e \in E}$ such
that $\hat{\ell}_e(x) \ge \ell_e(x)$ for all $x \ge 0$.
The goal of this modification is to create equilibria $\s$ with good
social cost denoted
by $\hat{C}_{\bar r}(\s)$, in the modified game with the {\em
  increased} latency functions, compared to the optimal social cost in
the original game using the {\em original}
latencies, when the unknown actual demand $r$ arrives. Adopting the standard practice of the learning augmented framework, we evaluate the performance of a mechanism in two different regimes of $r$: when the demand is accurate, $r=\bar{r}$, and  when $r$ is arbitrarily inaccurate.

We define the {\em Consistency} of the mechanism as the social cost of
the prediction-derived potential minimizer $\spred$ used to define the modified
latencies, compared to the cost of the optimal solution
(w.r.t. the original latencies) when the prediction is accurate
($r=\bar r$), and we take the worst-case $\bar r$. We define the
{\em Robustness} of a mechanism as the worst-case guarantee of a
potential minimizer $\spot$\footnote{Technically the potential minimizer
  may not be unique, in which case we consider the worst-case one.}
when the predicted demand may be {\em arbitrarily inaccurate}, taking
the worst-case over $r$ and $\bar r$. Formally,
\[
\textbf{(Consistency) } 
\max_{\bar r}
\frac{\hat{C}_{\bar r}(\spred)}{C_{\mathrm{opt}}(\bar r)}\,,
\qquad
\textbf{(Robustness) } \max_{r,\bar{r} }\max_{\spot \in \text{Pot}(r)}\frac{\hat{C}_{\bar r}(\spot)}{C_{opt}(r)}.\]

where Pot$(r)$ denotes the set of all potential minimizers of $\hat{\Phi}:=\Phi^{\hat{\ell}}$ w.r.t. the modified latencies for demand equal to $r$.
We also define a stronger notion of robustness, which we
refer to as {\em PoA}-Robustness, as the worst-case
guarantee for {\em any} equilibrium when the predicted demand may be
arbitrarily inaccurate\footnote{The PoA-Robustness guarantee holds for {\em
  any} actual demand and coincides with the {\em engineered Price of Anarchy} of 
coordination mechanisms~\cite{algorithmica/ChristodoulouMP14}.}
 \[\textbf{(PoA-Robustness) } \max_{r,\bar r} \max_{\s \in \text{Eq}(r)}\frac{\hat{C}_{\bar r}(\s)}{C_{opt}(r)}\,,\]
where Eq$(r)$ denotes the set of all equilibria under demand $r$.
We also define the approximate versions of these concepts. A learning-augmented mechanism is $\alpha$-consistent if its consistency is at most $\alpha$, and $\rho$-robust if its robustness is at most $\rho$.

\section{Non-atomic Congestion Games}
\label{sec:nonAtomic}

In this section, we study learning-augmented coordination mechanisms
for the class of non-atomic congestion games with polynomial
latencies of maximum degree $d$. We are given an instance augmented with a prediction of the
demand $(G,\bar{r},\ell)$, where $\bar{r}$ denotes the predicted
demand $\bar{r}=(\bar{r}_i)_{i\in \bar{\mathcal{I}}}$, and
$\bar{\mathcal{I}}$ is a prediction of the set of $\bar{k}$
commodities
$(\bar{s}_1,\bar{t}_1),\ldots, (\bar{s}_{\bar k},\bar{t}_{\bar k})$.

We propose a learning-augmented mechanism, which we call \textsc{Potential} (Definition~\ref{def:potentialMech}). Given $\bar{r}$, it
appropriately modifies the latency functions, achieving 
a trade-off between consistency and robustness which is tuned by a confidence parameter $\beta\in [1,d+1]$. The following theorem, which is the main
result of this section, summarizes the consistency--robustness trade-off
as a function of $\beta$.

\begin{restatable}{theorem}{consRobustThm}
\label{thm:consrobust}
Given the predicted input rates $\bar r$ and a parameter $\beta\in [1,d+1]$, the \textsc{Potential($\beta$,$\bar r$)} mechanism has  
    \begin{align*}
        \mbox{Consistency: }& \left(\beta\left(1-\frac{d}{d+1}\left(\frac{\beta}{d+1}\right)^{1/d}\right)\right)^{-1}\,;\qquad
        \mbox{Robustness: } \beta(d+1)\,.
    \end{align*}  
\end{restatable}

The consistency guarantee is evaluated on the prediction-derived potential minimizer used
to define the modified latencies, while the robustness guarantee holds for every
potential minimizer of the modified game. Increasing $\beta$ improves the consistency guarantee but worsens the robustness guarantee linearly in $\beta$. At $\beta=1$, the mechanism attains its strongest robustness guarantee, namely $d+1$, which is asymptotically a factor of $\Theta(\ln d)$ away from the PoS/PoA bound for non-atomic congestion games with polynomial latencies. At $\beta=d+1$, the mechanism is $1$-consistent, while its robustness guarantee becomes $(d+1)^2$. The above theorem gives the following corollary for the interesting and well-studied case of affine
latency functions, where $d=1$.

\begin{restatable}{corollary}{consrobustAffineCor}
\label{cor:consrobustAffine}
    Given the predicted input rates $\bar r$ and a parameter $\beta\in [1,2]$, the \textsc{Po\-ten\-tial($\beta$,$\bar r$)} mechanism for affine latency functions is $\frac{4}{(4-\beta)\beta}$-consistent and $2\beta$-robust.
\end{restatable}

We additionally provide a stronger PoA-Robustness guarantee that applies to the worst-case social cost across {\em all} equilibria, not just to potential minimizers.

\begin{restatable}{theorem}{consrobustPoAThm}
\label{thm:consrobustPoA}

Given the predicted input rates $\bar r$ and a parameter $\beta\in [1,d+1]$, the \textsc{Potential($\beta$,$\bar r$)} mechanism has  
    \begin{align*}
        \mbox{PoA-Robustness: }&
            \begin{cases}
                \left(1-\frac{\beta d}{d+1}\left(\frac{\beta}{d+1}\right)^{1/d}\right)^{-1}, \quad \mbox{for }\beta \in [1, (d+1)^{\frac 1{d+1}})\\
                \beta^{d+1}, \qquad \qquad \qquad \qquad\quad\;\; \mbox{for } \beta \in [(d+1)^{\frac 1{d+1}}, d+1]
            \end{cases}
    \end{align*}  
\end{restatable}

For the special case of affine latency functions, we obtain the following result, which serves as a warm-up for the proof of Theorem~\ref{thm:consrobustPoA}.

\begin{restatable}{corollary}{consrobustAffinePoACor}
\label{cor:consrobustAffinePoA}
    Given the predicted input rates $\bar r$ and a parameter $\beta\in [1,2]$, the \textsc{Po\-ten\-tial($\beta$,$\bar r$)} mechanism for affine latency functions has PoA-Robustness $\frac{4}{4-\beta^2}$ for $\beta \le \sqrt{2}$, and $\beta^2$ otherwise.
\end{restatable}

Two cases of the parameter $\beta$ are of special interest; when
$\beta = d+1$, the consistency bound becomes $1$ and the PoA-Robustness bound becomes
$(d+1)^{d+1}$, while for $\beta=1$, the consistency and PoA-Robustness bounds are both
equal to
$\left(1-\frac{d}{d+1}\left(\frac{1}{d+1}\right)^{1/d}\right)^{-1},$
 which matches the standard PoA bound of the original game {\em without predictions}.

The structure of this section is as follows. In Section~\ref{sec:threshold},
we introduce the general toolbox to design coordination mechanisms based on predicted rates, including the \textsc{Potential} mechanism and the main technical lemmas
used in its analysis. Section~\ref{sec:Affine} then provides, as a warmup, the analysis of \textsc{Potential} for the case of affine latency functions (Corollary~\ref{cor:consrobustAffine}). Next, in
Section~\ref{sec:Polynomials}, we extend the analysis of \textsc{Potential} for bounded-degree polynomials (Theorem~\ref{thm:consrobust}).
Finally, in Section~\ref{sec:nonAtomicPoA}, we establish PoA-Robustness guarantees for both affine (Corollary~\ref{cor:consrobustAffinePoA}) and
polynomial latency functions (Theorem~\ref{thm:consrobustPoA}).

\subsection{The {\sc Potential} Mechanism}
\label{sec:threshold}

We now present the main machinery for designing learning-augmented mechanisms
for non-atomic congestion games. We first introduce a generic coordination
mechanism, called \textsc{Threshold}, which takes as input a threshold flow
$f$ and a set of latency functions $h_e$, one for each resource $e\in E$.
We then state two useful properties of \textsc{Threshold}, summarized in
Lemmas~\ref{lemma:enforce-nonatomic-unique} and~\ref{lem:LikeVariationalIneq}.
Finally, we define the main mechanism of this section, called
\textsc{Potential}. This learning-augmented mechanism takes as input an
instance augmented with a demand prediction $(G,\bar r,\ell)$, together with a
confidence parameter $\beta$. Based on the predicted demand $\bar r$, it
modifies the latencies to $\hat{\ell}$, and the resulting equilibria are
evaluated on the actual instance $(G,r,\hat{\ell})$.

We begin with the formal definition of the \textsc{Threshold}
mechanism. Intuitively, \textsc{Threshold} is triggered once the
actual flow {\em exceeds some threshold} $f_e$, and modifies the
latency of this edge according to $h_e$. If the actual flow does not exceed the threshold, it uses the original latency function. Since
coordination mechanisms with continuous modifications have been shown
to yield no improvement~\cite{algorithmica/ChristodoulouMP14}, we
employ discontinuous latencies instead.

\begin{definition}[\textsc{Threshold}$(h,f)$] \label{def:simpleCM} 
    Given a set of latency functions $h=(h_e)_{e \in E}$ and a ``threshold'' flow $f=(f_e)_{e\in E}$, such that $h_e(x) \ge \ell_e(x)$ for all $x \geq 0$, we define the coordination mechanism \textsc{Threshold}$(h,f)$ with the following modified latency functions:
        \[
        \hat{\ell}_e(x)=
        \begin{cases}
            h_e(x), & x > f_e,\\
            \ell_e(x), & \text{otherwise}.
        \end{cases}
        \]
\end{definition}

An important property of the \textsc{Threshold} mechanism is that, when the
threshold flow is chosen as a minimizer of the potential $\Phi^h$, this same
flow is also a potential minimizer with respect to the modified latency
functions $\hat{\ell}$. Moreover, this particular flow preserves its original
social cost, since it never triggers the increased latencies. This property is
established in the following lemma.

\begin{lemma}
\label{lemma:enforce-nonatomic-unique}
    Given a predicted input rate $\bar r$ and a set of non-decreasing latency functions $(h_e)_{e \in E}$, with $\ell_e(x) \le h_e(x)$ for all $x \ge 0$ and $e \in E$, let $\fpred$ be a minimizer of $\Phi^h$ for $r = \bar r$, and let $\hat \ell$ be the modified latencies produced by  \textsc{Threshold}$(h,\fpred)$.
    Then $\fpred$ is a minimizer of $\Phi^{\hat{\ell}}$ for demand $\bar r$ and preserves the original social cost, i.e., $ \hat C(\fpred)=C(\fpred)$.
\end{lemma}

\begin{proof}
    Let $\hat\Phi$, $\Phi^h$, and $\Phi^{\ell}$ be the potentials induced by the latencies $\hat\ell$, $h$, and $\ell$, respectively. Let $\f$ be any feasible flow for demand $\bar r$. Define $E_h$ to be the set of edges for which the increased latencies $h$ are triggered, i.e., $f_e> \fpred_e$, and let $E_\ell = E\setminus E_h $ be the rest of the edges, for which the original latencies are used. The potential of $\f$ is bounded as follows:
    
\begin{align*}
\hat\Phi({\f}) 
&\geq \hat\Phi({\f}) +\Phi^h(\fpred) - \Phi^h(\f)\\
&=\sum_{e\in E_h}\left(\hat\Phi_e({\f}) +\Phi_e^h(\fpred) - \Phi_e^h(\f)\right)+ \sum_{e\in E_\ell}\left(\hat\Phi_e({\f}) +\Phi_e^h(\fpred) - \Phi_e^h(\f)\right)\\
&=\sum_{e\in E_h}\left(\left(\int_{0}^{\fpred_e}\ell_e(t)dt +\int_{\fpred_e}^{\f_e}h_e(t)dt\right) - \int_{\fpred_e}^{\f_e}h_e(t)dt \right)+ \sum_{e\in E_\ell}\left(\Phi_e^{\ell}({\f}) +\Phi_e^h(\fpred) - \Phi_e^h(\f)\right)\\
&= \sum_{e\in E_h}\Phi^{\ell}_e(\fpred)+\sum_{e\in E_\ell}\left(\Phi^{\ell}_e(\f)+ \int_{\f_e}^{\fpred_e}h_e(t)dt\right)\\
&= \sum_{e\in E_h}\Phi^{\ell}_e(\fpred)+ \sum_{e\in E_\ell}\left( \int_0^{\f_e}\ell_e(t)dt+\int_{\f_e}^{\fpred_e}h_e(t)dt\right)\\
&\geq\sum_{e\in E_h}\Phi^{\ell}_e(\fpred)+ \sum_{e\in E_\ell} \Phi^{\ell}_e(\fpred)\\
&= \hat\Phi({\fpred}), 
\end{align*}

where the first inequality is due to the fact that $\fpred$ minimizes $\Phi^h$, and the second due to $h_e(x)\geq \ell_e(x)$ for all $e\in E$ and $x\geq 0$. The last equality is due to the fact that the profile $\fpred$ never triggers the increased latencies $h$, hence $\hat{\Phi}(\fpred)= \Phi^\ell(\fpred)$.

The inequality above shows that $\hat\Phi(\f)\ge \hat\Phi(\fpred)$
for every feasible flow $\f$ for demand $\bar r$. Therefore, $\fpred$ is a minimizer of $\hat\Phi$ for demand $\bar r$. Moreover, since $\fpred$ never triggers the increased latencies, we also have $    \hat C(\fpred)=C(\fpred)$.
\end{proof}

The second important property of \textsc{Threshold} helps to provide robustness upper bounds.
The following lemma provides a variational inequality that
relates the social cost of any equilibrium w.r.t. the {\em
  modified} latencies to a benchmark flow under the auxiliary latencies $h$.

\begin{lemma}[Lifting Lemma] 
  \label{lem:LikeVariationalIneq}
  Given a set of continuous and non-decreasing
  latency functions $h=(h_e)_{e \in E}$, such that
  $\ell_e(x) \le h_e(x)$, for all $x\ge 0$ and $e \in E$, and a ``threshold'' flow $f^h=(f^h_e)_{e\in E}$, let
  $\hat{\ell}$ be the modified latencies of \textsc{Threshold}$(h,f^h)$. For any demand $r$, let
  $f$ be an equilibrium flow under $\hat{\ell}$ and
  $f^*$ be any other feasible flow. Then, 
\begin{align*}
    \sum_{e \in E}\hat{\ell}_e(f_e)f_e & \le \sum_{e \in E} h_e(f_e)f_e^*\,.
\end{align*}
\end{lemma}

\begin{proof}
We prove the equivalent path-based inequality
    \begin{align*}
    \sum_{p \in \mathcal{P}}\hat{\ell}_p(f)f_p & \le \sum_{p \in \mathcal{P}} h_p(f)f_p^*\,.
\end{align*}

For each commodity $i$, let $p_i \in \arg\min_{p\in \mathcal{P}_i}h_p(f)$. Since $f$ is an equilibrium with respect to $\hat\ell$, then for any path $p \in \mathcal{P}_i$ with $f_{p} > 0$, it holds that $\hat{\ell}_{p}(f) \le \liminf_{\epsilon \downarrow 0}\hat{\ell}_{p_i}(f + \epsilon\boldsymbol{1}_{p_i} -\epsilon\boldsymbol{1}_{p})\leq h_{p_i}(f)$. Then,
\begin{align*}
    \sum_{p \in \mathcal{P}}\hat{\ell}_p(f)f_p & 
    \le \sum_{i\in \mathcal{I}}\sum_{p \in \mathcal{P}_i}h_{p_i}(f)f_p \\
    & = \sum_{i\in \mathcal{I}}h_{p_i}(f)\sum_{p \in \mathcal{P}_i}f_p \\
    & = \sum_{i\in \mathcal{I}}h_{p_i}(f)\sum_{p \in \mathcal{P}_i}f^*_p \\
    & \le \sum_{i\in \mathcal{I}} \sum_{p \in \mathcal{P}_i} h_p(f)f_p^*\\
    & = \sum_{p \in \mathcal{P}} h_p(f)f_p^*\,,
\end{align*}
where the second equality is due to the fact that $\sum_{p \in \mathcal{P}_i}f_p=\sum_{p \in \mathcal{P}_i}f^*_p=r_i$, for all $i\in \mathcal{I}$, and the last inequality is derived by the definition of $p_i$.
\end{proof}

We are now ready to define the \textsc{Potential} 
mechanism, that takes as input a predicted rate $\bar r$ as well as a
confidence parameter $\beta$. It is then defined as a specific instantiation of
the \textsc{Threshold} mechanism, for a carefully selected
latency function $h(\beta)$ parametrized by $\beta$. For the case of
polynomials, $\beta \in [1,d+1]$, where $d$ is the maximum degree over
all latency functions. We define the latency functions $h(\beta)$ according to the potential function of the best
$\beta$-approximate equilibrium provided in
\cite{algorithmica/ChristodoulouKS11}, and with threshold flow equal to the prediction-derived potential minimizer.

\begin{definition} [\textsc{Potential($\beta$,$\bar r$)}]
\label{def:potentialMech}
    Given a parameter $\beta \in [1,d+1]$, we define the polynomial latency function $h(\beta)=(h_e)_{e\in E}$ as, 
    \[h_e(x)=\sum_{k=0}^{\lfloor\beta\rfloor-1}(k+1)a_{e,k}x^k+\sum_{k=\lfloor\beta\rfloor}^d\beta a_{e,k}x^k\,.\] 
        Let $\fpred$ be a potential minimizer of $\Phi^h$ for the predicted demand $\bar r$. 
    Then, the \textsc{Potential}$(\beta,\bar r)$ mechanism is the learning-augmented mechanism that uses the modified latencies produced by     \textsc{Threshold}$(h(\beta),\fpred)$.
    
\end{definition}

We analyze the \textsc{Potential} mechanism w.r.t. consistency and robustness in the next sections, first for affine latencies as a warmup (Sec.~\ref{sec:Affine}) and then for bounded-degree polynomials with nonnegative coefficients (Sec~\ref{sec:Polynomials}).

\subsection{Warmup - Affine Latencies}
\label{sec:Affine}
In this section, we provide consistency and robustness analysis of the
{\sc{Potential}} mechanism for the case of affine latency functions of
the form $\ell_e(x)=a_ex+b_e$, where $a_e$ and $b_e$ are non-negative
coefficients, specific to each resource $e$. Let $\beta$ be the
confidence parameter and $\bar r$ be the predicted demand, then the
threshold function $h(\beta)$ used by the {\sc{Potential}}$(\beta,\bar {r})$ boils down to $h_e(x)=\beta a_ex+b_e$,
for all $e\in E$, and the respective potential to
\[\Phi^\beta(f) = \sum_{e\in E}\int_{x=0}^{f_e}\frac{h_e(x)}{\beta}dx = \sum_{e\in E}\left(\frac{1}{2} a_ef^2_e+\frac{1}{\beta}b_ef_e\right). \]
Let $\fpred$ be a potential minimizer of $\Phi^h$ for demand equal to $\bar r$. Then the modified latencies of {\sc{Potential}}$(\beta,\bar {r})$ are defined as follows
        \[\hat{\ell}_e(x)=
    \begin{cases}
\beta a_ex+b_e, &  x > \fpred_e\\
a_ex+b_e, & \text{otherwise}.
    \end{cases}\]

\consrobustAffineCor*

\begin{proof}

  \noindent{\bf Consistency.}  
    We first show the consistency bound. Suppose that the prediction is accurate,
  i.e., $\bar r=r$. Let $f^*$ be the optimal flow w.r.t. the original latencies
  $\ell$. By Lemma~\ref{lemma:enforce-nonatomic-unique},
  the potential minimizer $\fpred$ is also a minimizer of
  $\hat{\Phi}$ and preserves its original social cost, i.e., $      \hat C(\fpred)=C(\fpred)$. 
  Moreover, $\fpred$ minimizes $\Phi^h$. Equivalently, it minimizes
  $\Phi^h/\beta$, which is the approximate potential used in \cite{algorithmica/ChristodoulouKS11}. Therefore, by the bound of \cite{algorithmica/ChristodoulouKS11} for $\beta$-approximate Wardrop
  equilibria,\footnote{This approach was used in \cite{algorithmica/ChristodoulouKS11} to bound the PoS of
  $\beta$-approximate Wardrop equilibria.} $
      C(\fpred)
      \le
      \frac{4}{(4-\beta)\beta}C(f^*)$.
  Combining this with $\hat C(\fpred)=C(\fpred)$ yields $\hat C(\fpred)
      \le
      \frac{4}{(4-\beta)\beta}C(f^*)$,
  which gives the required consistency bound.

\noindent{\bf Robustness.}
Let $f$ be any feasible flow and let $E_h$ be the set of edges for which $f_e > \bar{f}_e$, i.e., $\hat{\ell}_e(f_e)=h_e(f_e)$, and let $E_\ell \subseteq E \setminus E_h$ be the rest of the edges, for which the original latencies are used. Then, $\hat{\Phi}(f)$ can be written as

\begin{align*}
\hat{\Phi}(f)
&=\sum_{e\in E} \int_0^{f_e}\hat{\ell}_e(t)dt\\
&=\sum_{e \in E_\ell}\int_0^{f_e}(a_et+b_e)dt
+\sum_{e\in E_h}\left(
    \int_0^{\bar f_e}(a_et+b_e)dt
    +\int_{\bar f_e}^{f_e}(\beta a_et+b_e)dt
\right) \\
&=\sum_{e \in E_\ell}\left(\frac{1}{2}a_ef_e^2+b_ef_e\right)
+\sum_{e\in E_h}\left(
    \frac{\beta}{2}a_ef_e^2+b_ef_e
    -\frac{\beta-1}{2}a_e\bar f_e^2
\right) \\
&\le \sum_{e \in E}a_ef_e^2+b_ef_e= C(f)\,,
\end{align*}
where the inequality uses $\beta\le 2$ and $\bar f_e< f_e$ for every
$e\in E_h$.

Moreover,
\begin{align*}
    \hat\Phi(f)
    &=
    \sum_{e \in E_\ell}\left(\frac{1}{2}a_ef_e^2+b_ef_e\right)
    +
    \sum_{e\in E_h}\left(
        \frac{\beta}{2}a_ef_e^2+b_ef_e
        -\frac{\beta-1}{2}a_e\bar f_e^2 \right
    ) \\
    &\ge
    \sum_{e \in E_\ell}\left(\frac{1}{2}a_ef_e^2+b_ef_e\right)
    +
    \sum_{e\in E_h}\left(\frac{1}{2}a_ef_e^2+b_ef_e\right) \\
    &\ge
    \frac{1}{2\beta}\left(
        \sum_{e\in E_\ell}(a_ef_e^2+b_ef_e)
        +
        \sum_{e\in E_h}(\beta a_ef_e^2+b_ef_e)
    \right) \\
    &=
    \frac{1}{2\beta}\hat C(f)\,.
\end{align*}
Combining the above inequalities yields $\frac{1}{2\beta}\hat{C}(\fpot)\le\hat{\Phi}(\fpot)\le \hat{\Phi}(f^*)\le C(f^*)$, since they hold for any $f$, and therefore, for $\fpot$ and $f^*$ too, while the middle inequality follows from the fact that $\fpot$ is a potential minimizer. This establishes the desired robustness bound and concludes the proof.
\end{proof}

\subsection{Bounded-degree Polynomials} 
\label{sec:Polynomials}

We now show how to extend our analysis to the general case where latency functions are polynomials of maximum degree $d$. Let $\beta \in [1,d+1]$.
The consistency argument remains unchanged, hence,
it is equal to the PoS of $\beta$-approximate equilibria in non-atomic congestion games for polynomial latency functions. Regarding robustness, we bound the inefficiency of potential minimizers to get a guarantee for the PoS for the modified game.

\consRobustThm*

\begin{proof}
\noindent{\bf Consistency.}
    We first establish the consistency upper bound. Similarly to the proof for the affine latency functions, we obtain that the consistency is equal to the PoS of the best $\beta$-approximate equilibrium. Suppose that the prediction is
  accurate, i.e., $\bar r=r$.  The \textsc{Potential($\beta$,$\bar r$)}
  mechanism modifies the latencies according to \textsc{Threshold}$(h,\fpred)$, for
  \[
h_e(x)=
\sum_{k=0}^{\lfloor\beta\rfloor-1}(k+1)a_{e,k}x^k
+
\sum_{k=\lfloor\beta\rfloor}^d\beta a_{e,k}x^k\,,
\qquad \forall e\in E,
\]
and $\fpred$ is a potential minimizer of $\Phi^h$ for input rate
$\bar r=r$. By Lemma~\ref{lemma:enforce-nonatomic-unique}, $\fpred$ is a minimizer of
$\hat{\Phi}=\Phi^{\hat\ell}$ and preserves its original social cost, i.e., $    \hat C(\fpred)=C(\fpred)$. Moreover, $\fpred$ minimizes $\Phi^h$. In \cite{algorithmica/ChristodoulouKS11} it is shown that any flow $f$ that minimizes the approximate potential 
  \[\Phi^\beta(f) = \sum_{e\in E}\int_{0}^{f_e}\frac{h_e(x)}{\beta}dx \]
  has social cost w.r.t. the original latencies $\ell$ at most a factor of 
  $\alpha = \left(\beta\left(1-\frac{d}{d+1}\left(\frac{\beta}{d+1}\right)^{1/d}\right)\right)^{-1}$ away from the optimal cost.  Note that $\fpred$ also minimizes $\Phi^\beta=\Phi^h/\beta$, and since $\hat{C}(\fpred)=C(\fpred)$, the mechanism 
 is $\alpha$-consistent.

\noindent{\bf Robustness.}
Let $f$ be any feasible flow and let $E_h$ be the set of edges for which $f_e>\fpred_e$, i.e. $\hat{\ell}_e(f_e)=h_e(f_e)$ and let $E_\ell = E\setminus E_h $ be the rest of the edges, where the original latencies are used. Then $\hat{\Phi}(f)$ can be written as
\[
\hat{\Phi}(f) = \sum_{e\in E} \int_{0}^{f_e}\hat{\ell}_e(t)dt 
=\sum_{e\in E_{\ell}} \int_{0}^{f_e}\ell_e(t)dt  + \sum_{e\in E_h} \left(\int_0^{\bar f_e} \ell_e(t) dt + \int_{\bar f_e}^{f_e}h_e(t)dt \right)\,. 
\]
We split the rest of the proof into two claims. We first show that $\hat{\Phi}(f)\leq C(f)$ and then that $\hat{\Phi}(f)\geq \frac{1}{\beta(d+1)}\hat{C}(f)$. We finally combine them to get the robustness bound.

\begin{claim}
    $\hat{\Phi}(f)\leq C(f)\,$.
\end{claim}
\begin{proof}
The part of the potential containing edges in $E_h$ can be written as follows 
\begin{align*}
\sum_{e\in E_h}\left(\int_0^{\bar f_e}\ell_e(t)dt+\int_{\bar f_e}^{f_e}h_e(t)dt\right) & \le 
\sum_{e \in E_h}\int_{0}^{f_e}h_e(t)dt \\
& = \sum_{e \in E_h}\sum_{k=0}^{\lfloor\beta\rfloor-1}(k+1)\int_{0}^{f_e}a_{e,k}t^kdt + \sum_{e \in E_h}\sum_{k=\lfloor\beta\rfloor}^{d}\beta\int_{0}^{f_e}a_{e,k}t^kdt \\
& = \sum_{e \in E_h}\sum_{k=0}^{\lfloor\beta\rfloor-1}a_{e,k}f_e^{k+1} + \sum_{e \in E_h}\sum_{k=\lfloor\beta\rfloor}^{d}\frac{\beta}{k+1} a_{e,k}f_e^{k+1} \\
& \leq \sum_{e \in E_h}\sum_{k=0}^{d}a_{e,k}f_e^{k+1}
\end{align*}

Note that for every index $k \ge \lfloor \beta \rfloor$, we have $k+1>\beta$, hence $\frac{\beta}{k+1} \le 1$, which justifies the inequality above. For the edges in $E_{\ell},$ it holds that:

\begin{align*}
    \sum_{e \in E_{\ell}}\int_{0}^{f_e}\ell_e(t)dt = \sum_{e \in E_{\ell}}\sum_{k=0}^{d}\frac{a_{e,k}}{k+1} f_e^{k+1} \leq \sum_{e \in E_{\ell}}\sum_{k=0}^{d}a_{e,k}f_e^{k+1} 
\end{align*}

Overall, the two inequalities above give

\begin{align*}
\hat{\Phi}(f) \leq \sum_{e \in E}\sum_{k=0}^{d}a_{e,k}f_e^{k+1} = C(f)\,.
\end{align*}
\end{proof}

\begin{claim}
    $\hat{\Phi}(f)\geq \frac{1}{\beta(d+1)}\hat{C}(f)\,.$
\end{claim}
\begin{proof}
We have that 
\begin{align*}
    \hat{\Phi}(f)
    &\ge
    \sum_{e \in E_\ell} \sum_{k=0}^d
        \frac{1}{\beta(d+1)}a_{e,k}f_e^{k+1} \\    &\quad+
    \sum_{e\in E_h}
    \left(
        \sum_{k=0}^{\lfloor\beta\rfloor-1}
        \frac{1}{\beta(d+1)}(k+1)a_{e,k}f_e^{k+1}
        +
        \sum_{k=\lfloor\beta\rfloor}^d
        \frac{1}{\beta(d+1)}\beta a_{e,k}f_e^{k+1}
    \right) \\
    &=
    \frac{1}{\beta(d+1)}
    \left(
        \sum_{e\in E_\ell}\sum_{k=0}^d a_{e,k}f_e^{k+1}
        +
        \sum_{e\in E_h}
        \left(
            \sum_{k=0}^{\lfloor\beta\rfloor-1}(k+1)a_{e,k}f_e^{k+1}
            +
            \sum_{k=\lfloor\beta\rfloor}^d \beta a_{e,k}f_e^{k+1}
        \right)
    \right) \\
    &=
    \frac{1}{\beta(d+1)}\hat{C}(f)\,.
\end{align*}
The inequality uses $\bar f_e < f_e$ for every $e\in E_h$, $\beta\ge 1$, and
$k+1\le d+1$ for every $k\in\{0,\ldots,d\}$.
\end{proof}

Let $\fpot$ be any potential minimizer of $\hat{\Phi}$, and $f^*$ be an optimal profile w.r.t. the original latency functions, i.e., $f^*$ minimizes $C(f)$. Then, by utilizing the above inequalities that hold for any $f$, and therefore for $\fpot$ and $f^*$, it follows that
$\frac{1}{\beta(d+1)}\hat{C}(\fpot) \leq \hat{\Phi}(\fpot)\leq \hat{\Phi}(f^*) \leq C(f^*)\,,$
which concludes the proof for robustness, and so completes the proof of the theorem.
\end{proof}

\subsubsection{A threshold-sensitive robustness refinement}
\label{sec:threshold-sensitive-nonatomic}

We present a refined threshold-sensitive analysis of the \textsc{Potential} mechanism, depending on how close the triggered edges are to their predicted thresholds. Let \(f\) be any feasible flow and let
\(E_h(f)=\{e\in E:f_e>\bar f_e\}\) be the set of triggered edges. For every \(e\in E_h(f)\), define
\(\eta_e=\bar f_e/f_e\in[0,1)\). If \(E_h(f)\neq\emptyset\), let
\(\eta(f,\bar f)=\max_{e\in E_h(f)}\eta_e\), and otherwise set \(\eta(f,\bar f)=0\). Recall that
\(\hat\Phi_e(f)=\int_0^{f_e}\hat\ell_e(t)\,dt\), \(\hat\Phi(f)=\sum_{e\in E}\hat\Phi_e(f)\), and
\(\hat C_e(f)=f_e\hat\ell_e(f_e)\).

We will use two elementary facts. First, \(h_e(x)\le \beta\ell_e(x)\) for every edge \(e\) and every \(x\ge0\). Indeed, by the definition of \(h_e\), each monomial coefficient is either multiplied by \(\beta\), or by \(k+1\) with \(k+1\le\beta\); all coefficients are nonnegative. Second, since \(h_e\) is a polynomial of degree at most \(d\) with nonnegative coefficients, writing \(h_e(x)=\sum_{k=0}^d b_{e,k}x^k\) gives
\(\Phi^h_e(x)=\int_0^x h_e(t)\,dt=\sum_{k=0}^d b_{e,k}x^{k+1}/(k+1)\ge xh_e(x)/(d+1)\).

\begin{lemma}[Threshold-sensitive potential lower bound]
\label{lemma:threshold-sensitive-nonatomic}
For every feasible flow \(f\),
\[
    \hat{\Phi}(f)
    \ge
    \frac{1-\left(1-\frac{1}{\beta}\right)\eta(f,\bar f)}
    {d+1}\,
    \hat{C}(f).
\]
\end{lemma}

\begin{proof}
We prove the inequality edge by edge. First consider an untriggered edge \(e\notin E_h(f)\), so \(f_e\le \bar f_e\). Then \(\hat\ell_e(f_e)=\ell_e(f_e)\), and hence
\(\hat\Phi_e(f)=\int_0^{f_e}\ell_e(t)\,dt\ge f_e\ell_e(f_e)/(d+1)=\hat C_e(f)/(d+1)\), where the inequality uses that \(\ell_e\) is a polynomial of degree at most \(d\) with nonnegative coefficients. Since
\(1-(1-1/\beta)\eta(f,\bar f)\le 1\), this already implies the desired edge-wise bound for every untriggered edge.

Now consider a triggered edge \(e\in E_h(f)\), so \(f_e>\bar f_e\). Then
\begin{align*}
    \hat\Phi_e(f)
    &=
    \int_0^{\bar f_e}\ell_e(t)\,dt
    +
    \int_{\bar f_e}^{f_e}h_e(t)\,dt  \\
    &=
    \Phi^\ell_e(\bar f_e)+\Phi^h_e(f_e)-\Phi^h_e(\bar f_e) \\
    &\ge
    \Phi^h_e(f_e)
    -
    \left(1-\frac{1}{\beta}\right)\Phi^h_e(\bar f_e) \\
    &\ge
    \Phi^h_e(f_e)
    -
    \left(1-\frac{1}{\beta}\right)\eta_e\Phi^h_e(f_e) \\
    &=
    \left(1-\left(1-\frac{1}{\beta}\right)\eta_e\right)\Phi^h_e(f_e) \\
    &\ge
    \frac{1-\left(1-\frac{1}{\beta}\right)\eta_e}{d+1}f_eh_e(f_e) \\
    &=
    \frac{1-\left(1-\frac{1}{\beta}\right)\eta_e}{d+1}\hat C_e(f).
\end{align*}
The first inequality uses \(h_e(t)\le \beta\ell_e(t)\), which implies
\(\Phi^\ell_e(\bar f_e)\ge \Phi^h_e(\bar f_e)/\beta\). The second inequality
uses \(\Phi^h_e(\bar f_e)\le \eta_e\Phi^h_e(f_e)\): indeed,
\(\Phi^h_e\) has nonnegative coefficients and zero constant term, and
\((\bar f_e/f_e)^q\le \bar f_e/f_e=\eta_e\) for every \(q\ge1\). The last
inequality uses \(\Phi^h_e(f_e)\ge f_eh_e(f_e)/(d+1)\). Since
\(\eta_e\le \eta(f,\bar f)\) and \(1-(1-1/\beta)\eta\) is decreasing in
\(\eta\), this gives
\[
    \hat\Phi_e(f)
    \ge
    \frac{1-\left(1-\frac{1}{\beta}\right)\eta(f,\bar f)}
    {d+1}\hat C_e(f)
\]
for every triggered edge. Summing over all edges completes the proof.
\end{proof}

\begin{corollary}[Threshold-sensitive robustness]
\label{cor:threshold-sensitive-robustness}
Let \(f^\diamond\) be any potential minimizer of the modified game, and let
\(f^\ast\) be an optimal flow under the original latencies. Then
\[
    \hat C(f^\diamond)
    \le
    \frac{d+1}
    {1-\left(1-\frac{1}{\beta}\right)\eta(f^\diamond,\bar f)}
    C(f^\ast).
\]
\end{corollary}

\begin{proof}
By Lemma~\ref{lemma:threshold-sensitive-nonatomic},
\[
    \frac{1-\left(1-\frac{1}{\beta}\right)\eta(f^\diamond,\bar f)}
    {d+1}
    \hat C(f^\diamond)
    \le
    \hat\Phi(f^\diamond).
\]
Since \(f^\diamond\) minimizes \(\hat\Phi\), and since the proof of
Theorem~\ref{thm:consrobust} shows that \(\hat\Phi(f)\le C(f)\) for every feasible flow \(f\), we have
\(\hat\Phi(f^\diamond)\le \hat\Phi(f^\ast)\le C(f^\ast)\). Combining the two inequalities and rearranging gives the claim.
\end{proof}

This refinement interpolates between the strongest robustness bound \(d+1\), obtained when no triggered edge is close to its threshold, and the worst-case guarantee \(\beta(d+1)\) of Theorem~\ref{thm:consrobust} as \(\eta(f^\diamond,\bar f)\) approaches \(1\). Thus, the robustness loss depends not only on the confidence parameter \(\beta\), but also on how close the realized potential minimizer is to the prediction-derived threshold flow on the edges where the mechanism is triggered.

\subsection{PoA-Robustness}\label{sec:nonAtomicPoA}

In this section, we analyze the PoA-Robustness of \textsc{Potential}, first for the case of affine latencies (Corollary~\ref{cor:consrobustAffinePoA}) and subsequently for the case of polynomial latencies (Theorem~\ref{thm:consrobustPoA}). 

\consrobustAffinePoACor*

\begin{proof}    
Given any input rate $r$, let $f$ be an equilibrium flow under the modified latencies $\hat{\ell}$ and let $f^*$ be the optimal flow under the original latencies $\ell$. By using the Lifting Lemma (Lemma~\ref{lem:LikeVariationalIneq}) and substituting for $h_e(x)=\beta a_ex+b_e$, we obtain
\begin{align}
\hat{C}(f)= \sum_{e \in E}\hat{\ell}_e(f_e)f_e \le \sum_{e \in E} h_e(f_e)f_e^* = \sum_{e\in E} (\beta a_ef_ef_e^*+b_ef_e^*)\,. \label{eq:relationToCKS11}
\end{align}

A quick ---but not tight---  upper bound of PoA-Robustness can be obtained as follows. 
The last term of (\ref{eq:relationToCKS11}) can be bounded above by  $\beta \sum_{e \in E}\ell_e(f_e)f^*_e$, which gives $\hat{C}(\f)\leq \beta \sum_{e \in E}\ell_e(f_e)f^*_e.$
By observing that the above inequality is also satisfied by $\beta$-approximate equilibria w.r.t. the original latencies, one can follow a similar analysis with that of \cite{algorithmica/ChristodoulouKS11} to upper bound the PoA of $\beta$-approximate equilibria. This gives a PoA-Robustness bound of at most $\frac{4\beta}{4-\beta}$, for $\beta\in [1,2]$.

However, we are able to obtain an improved bound for all $\beta \in (1,2)$ (see Figure \ref{fig:robustness_affine}.(a) for a comparison) by a more careful analysis. For this, we need to establish the following claim.

\begin{claim}
    For any parameter $\lambda \geq  1/\beta$ it holds that 
    \begin{align*}
    \hat{C}(f) \leq \frac{\beta}{4\lambda}\hat{C}(f)+\lambda \beta C(f^*) \,.
\end{align*}
\end{claim}

\begin{proof}
    By the first inequality of \eqref{eq:relationToCKS11} and the fact that for any real numbers $x,y,\lambda$ it holds that $xy\leq \frac 1{4\lambda}x^2+\lambda y^2 $ (Lemma 2 in \cite{arxiv:CKS08}), we get
\begin{align*}
    \hat{C}(f) & \le \sum_{e\in E} (\beta a_ef_ef_e^*+b_ef_e^*)\\
    &\le \sum_{e\in E} \left(\beta a_e\left(\frac 1{4\lambda}(f_e)^2+\lambda  (f_e^*)^2\right)+b_ef_e^*\right) \\
    & \leq \frac{\beta}{4\lambda}\sum_{e\in E} (a_e (f_e)^2+b_ef_e)+\lambda \beta \sum_{e\in E} (a_e (f_e^*)^2+b_ef_e^*) \mbox{\qquad (since $\lambda \beta \ge 1$)}\\
    &\leq \frac{\beta}{4\lambda}\hat{C}(f)+\lambda\beta C(f^*) \,.
\end{align*}
\end{proof}

Rearranging yields the following bound on the robustness:
\begin{align*}
    \frac{\hat{C}(f)}{C(f^*)} & \le \frac{\lambda\beta}{1-\frac{\beta}{4\lambda}} = \frac{4\lambda^2\beta}{4\lambda-\beta}\,.
\end{align*}

We appropriately choose $\lambda \geq 1/\beta$ to minimize the above
expression, which is $\lambda=1/\beta$ for $\beta\leq \sqrt{2}$, and
$\lambda=\beta/2$ otherwise. This overall results in robustness of
$\frac 4{4-\beta^2}$ for $\beta\leq \sqrt{2}$, and robustness of
$\beta^2$ otherwise.  We remark that this bound works for \emph{any} equilibrium
flow under $\hat\ell$.
\end{proof}

 \begin{figure}[h]
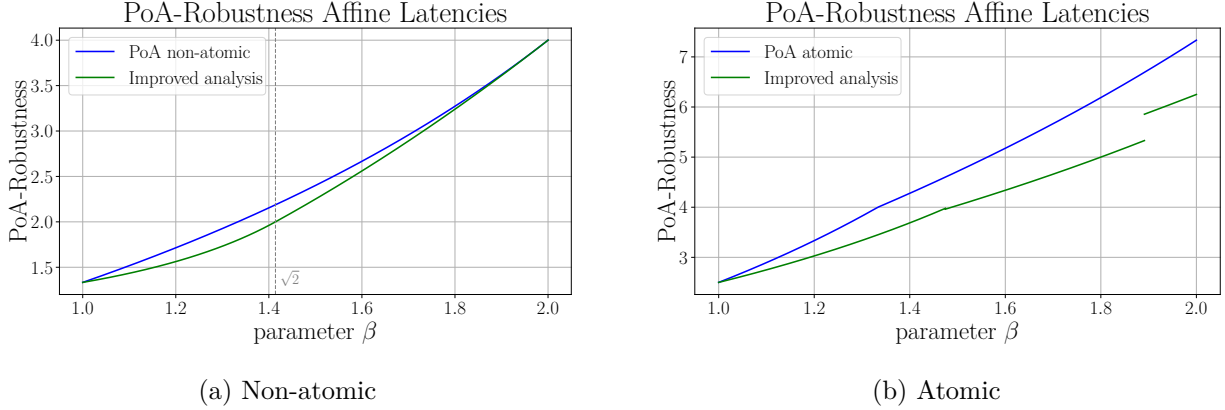

     \centering

     \begin{subfigure}[t]{0.48\textwidth}
         \centering
         \resizebox{\linewidth}{!}{\input{figs/robustness_non-atomic.pgf}}
         \caption{Non-atomic}
     \end{subfigure}
     \hfill
     \begin{subfigure}[t]{0.48\textwidth}
         \centering
         \resizebox{\linewidth}{!}{\input{figs/robustness_atomic.pgf}}
         \caption{Atomic}
     \end{subfigure}

     \caption{PoA-Robustness guarantees for affine latency functions. The left panel corresponds to the non-atomic setting and the right one to the atomic setting. In green, we highlight the improved bounds obtained by our analysis, compared to the PoA for $\beta$-approximate equilibria.}
     \label{fig:robustness_affine}
 \end{figure}

\consrobustPoAThm*

\begin{proof}   

Given any input rate $r$, suppose that $f$ is the equilibrium flow under the modified latencies $\hat{\ell}$, and $f^*$ is the optimal flow under the original latencies $\ell$. By the Lifting Lemma~\ref{lem:LikeVariationalIneq} we get: 
\begin{align*}
\hat C(f) = \sum_{e \in E}\hat{\ell}_e(f_e)f_e \le \sum_{e\in E} h_e(f_e)f_e^*\,.
\end{align*}

Similarly to the affine latency case, the last term can be bounded above by $\beta \sum_{e \in E}\ell_e(f_e)f^*_e$, which results in $\hat C(f) \le\beta \sum_{e \in E}\ell_e(f_e)f^*_e$. One can thus follow the analysis to upper bound the PoA of $\beta$-approximate equilibria provided by \cite{algorithmica/ChristodoulouKS11} and obtain a PoA-Robustness bound of at most 
$\left(\frac{1}{\beta}-\frac{d}{(d+1)^{1+\frac{1}{d}}}\right)^{-1}$ for $\beta \leq \sqrt[d]{d+1}$, and $\beta^{d+1}$ for $\sqrt[d]{d+1}<\beta \le d+1$.
However, we are able to provide a tighter analysis for our setting that improves this bound for all $\beta \in (1,\sqrt[d]{d+1})$. For this, we show the following claim:

\begin{claim}
    For any parameter $\lambda \geq 1/\beta$ it holds that 
    \begin{align*}
    \hat{C}(f) \leq \max\left\{\frac{1}{4\lambda},\frac{\beta d}{(d+1)\sqrt[d]{\lambda(d+1)}}\right\}\hat{C}(f)+ \lambda\beta C(f^*) \,.
\end{align*}
\end{claim}

\begin{proof}

First, consider the inequality stated in Lemma~\ref{lem:arithmeticNonAtomic} (see Appendix~\ref{appendix:Arithmetic}). Combining this result with Lemma~\ref{lem:LikeVariationalIneq}, and noting that $\lambda \beta \geq 1$, we obtain
    
\begin{align*}
    \hat{C}(f) &\leq \sum_{e\in E}h_e(f_e)f_e^* \\
    &= \sum_{e\in E}\left( \sum_{k=0}^{\lfloor\beta\rfloor-1}(k+1)a_{e,k}f_e^kf_e^*+\sum_{k=\lfloor\beta\rfloor}^d\beta a_{e,k}f_e^kf_e^*\right)\\
    &\le \sum_{e\in E}\left( a_{e,0}f_e^*+\sum_{k=1}^{\lfloor\beta\rfloor-1}\beta a_{e,k}f_e^kf_e^*+\sum_{k=\lfloor\beta\rfloor}^d\beta a_{e,k}f_e^kf_e^*\right) \\
    &\leq \sum_{e\in E}\left(\sum_{k=0}^{d}\beta a_{e,k}\max\left\{\frac{1}{4\lambda},\frac{d}{(d+1)\sqrt[d]{\lambda(d+1)}}\right\}f_e^{k+1} + \sum_{k=0}^{d}\beta a_{e,k}\lambda(f_e^*)^{k+1}\right) \\
    &\leq \beta \max\left\{\frac{1}{4\lambda},\frac{d}{(d+1)\sqrt[d]{\lambda(d+1)}}\right\}\hat{C}(f)+  \lambda \beta C(f^*) \,.
\end{align*}
\end{proof}
Rearranging yields the following bound on the PoA-Robustness:

\begin{align*}
    \frac{\hat{C}(f)}{C(f^*)} & \leq \frac{\lambda\beta}{1-\beta \max\left\{\frac{1}{4\lambda},\frac{d}{(d+1)\sqrt[d]{\lambda(d+1)}}\right\}}\,.
\end{align*}
For $d=1$, it holds that $\frac{1}{4\lambda} = \frac{d}{(d+1)\sqrt[d]{\lambda(d+1)}}$, and we derive the bounds of Corollary~\ref{cor:consrobustAffinePoA}. 

Suppose now that $d\geq 2$. Define $\lambda_0 = (\frac{1}{4d})^\frac{d}{d-1}(d+1)^{\frac{d+1}{d-1}}$ to be the value of $\lambda$ such that $\frac{1}{4\lambda} = \frac{d}{(d+1)\sqrt[d]{\lambda(d+1)}}$. Note that for $\lambda \leq \lambda_0$ it is $\frac{1}{4\lambda} \geq \frac{\beta d}{(d+1)\sqrt[d]{\lambda(d+1)}}$, and vice-versa for $\lambda > \lambda_0$.
We appropriately choose the value of $\lambda \geq 1/\beta$ to minimize the above expression. Observe that in case $\beta < \frac{1}{\lambda_0},$ it holds that $\lambda \geq \frac{1}{\beta} > \lambda_0$, and the above expression becomes $\frac{\lambda\beta}{1-\frac{\beta d}{(d+1)\sqrt[d]{\lambda(d+1)}}}.$ This is minimized for $\lambda = 1 / \beta$ when $\beta \in [1, (d+1)^{\frac 1{d+1}})$ and $\lambda = \frac{\beta^d}{d+1}$ when $\beta \in [(d+1)^{\frac 1{d+1}}, \frac{1}{\lambda_0})$. For $\beta \in [\frac{1}{\lambda_0}, d+1]$, it holds that $\beta \geq 1/\lambda_0 \geq 4\lambda_0$ and therefore if $\lambda \leq \lambda_0$ the PoA-Robustness is negative. Thus, again the expression is optimized for $\lambda = \frac{\beta^d}{d+1} \geq \lambda_0.$ Overall, this gives:

$$
\lambda =
\begin{cases}
\frac{1}{\beta}, \quad\quad \mbox{for } \beta \in [1, (d+1)^{\frac 1{d+1}})\\
\frac{\beta^d}{d+1}, \quad \mbox{for } \beta \in [(d+1)^{\frac 1{d+1}}, d+1]
\end{cases}
$$
In total, PoA-Robustness is bounded by:

$$
\frac{\hat{C}(f)}{C(f^*)} \leq
    \begin{cases}
        \left(1-\beta\frac{d}{d+1}\left(\frac{\beta}{d+1}\right)^{1/d}\right)^{-1}, \quad \mbox{for }\beta \in [1, (d+1)^{\frac 1{d+1}})\\
        \beta^{d+1}, \qquad \qquad \qquad \qquad\quad\;\;\; \mbox{for } \beta \in [(d+1)^{\frac 1{d+1}}, d+1)
    \end{cases}
$$
\end{proof}

\section{Atomic Congestion Games}
\label{sec:atomic}

In this section, we switch to the class of atomic congestion games
with polynomial latency functions. As in the 
non-atomic case, we propose the \textsc{Potential}
mechanism and analyze its performance. Given an instance of a
congestion game augmented with a prediction of the demand
$(\bar{N},E, (\bar{\mathcal{S}}_i)_{i\in \bar{N}},(\ell_e)_{e\in E}
)$, where $\bar{r} = (\bar{N}, (\bar{\mathcal{S}}_i)_{i\in \bar N})$
denotes the predicted demand, the \textsc{Potential} mechanism
modifies the latencies according to a \textsc{Threshold} coordination
mechanism.

The main result of this section offers a trade-off between consistency and robustness, tuned by a confidence parameter $\beta$, and is summarized in the next theorem.

\begin{restatable}{theorem}{atomicThm}
\label{thm:atomic}
    Given a predicted input demand $\bar r$ and a parameter $\beta\in \left[\frac{2(d+1)}{d+3},d+1\right]$, the  \textsc{Potential($\beta$,$\bar r$)} mechanism in atomic congestion games has   $$\mbox{Consistency: }\frac{d+1}{\beta}  \,; \qquad \mbox{Robustness: } \beta(d+1)\,.$$
\end{restatable}

For the case of affine latency functions, we are able to get an improved analysis for the consistency bound.

\begin{restatable}{corollary}{atomicAffineCor}
\label{cor:atomicAffine}
    Given the predicted input rates $\bar r$ and a parameter $\beta\in [1,2]$, the \textsc{Po\-ten\-tial($\beta$,$\bar r$)} mechanism is $\frac{\sqrt{3}+1}{\sqrt{3}+\beta-1}$-consistent, and $2\beta$-robust for affine latency functions.
\end{restatable}

As in the case of non-atomic congestion games, we additionally provide a stronger PoA-Robustness guarantee that applies to the worst-case
social cost across all equilibria, not only to potential minimizers.

\begin{restatable}{theorem}{atomicPoAThm}
\label{thm:atomicPoA}
    Given a predicted demand $\bar r$ and a parameter $\beta\in \left[{2(d+1)}/({d+3}),d+1\right]$, the PoA-Robustness of \textsc{Potential($\beta$,$\bar r$)} for atomic congestion games is at most $\phi_{d,\beta}^{d+1}$, where $\phi_{d,\beta}$ is the solution of $z^{d+1}=\beta (z+1)^d$.
\end{restatable}

For the special case of affine latency functions, we obtain a refined analysis.

\begin{restatable}{corollary}{atomicAffinePoACor}
\label{cor:atomicAffinePoA}

    Given the predicted input rates $\bar r$ and a parameter $\beta\in [1,2]$, the \textsc{Po\-ten\-tial($\beta$,$\bar r$)} mechanism has PoA-Robustness $\frac{\beta z^2+(\beta+2)z+\frac{\beta+3}{4}}{2z+1-\beta}$, for affine latency functions, where $z$ is the maximum integer that satisfies $z \leq  \frac{\beta - 1 + \sqrt{\beta^4+4\beta^2-1}}{2}$.   
    
\end{restatable}

For the case of affine latencies, in the extreme value of $\beta=2$, Corollary~\ref{cor:atomicAffine} gives an (optimal) consistency bound of $1$, and Corollary~\ref{cor:atomicAffinePoA} gives a PoA-Robustness of $6.25$, whereas for the other extreme value of $\beta = 1$, the consistency becomes approximately $1.577$ (which coincides with the PoS of exact equilibria), while PoA-Robustness is $5/2$
(which coincides with the PoA of exact equilibria).

\begin{remark}
We remark that the bounds of Corollaries \ref{cor:atomicAffine} and \ref{cor:atomicAffinePoA} provide better consistency and PoA-Robustness guarantees compared to the results of Theorems \ref{thm:atomic} and \ref{thm:atomicPoA} for $d=1$. Applying $d=1$ to Theorems \ref{thm:atomic} and \ref{thm:atomicPoA} gives a higher consistency bound of $2$ (instead of $\simeq 1.577$) and a higher PoA-Robustness bound of approximately $7.33$ (instead of $6.25$), for $\beta=1$ and $\beta=2$, respectively.
\end{remark}

The remainder of this section is organized as follows. Section
\ref{sec:potentialAtomic} introduces the definitions and lemmas used to design
learning-augmented coordination mechanisms. In Section \ref{sec:AffineAtomic}, we analyze the
performance of the \textsc{Potential} mechanism for the case of affine
latency functions (Corollary \ref{cor:atomicAffine}). Section \ref{sec:PolynomialsAtomic} extends the analysis to the more
general setting of bounded-degree polynomials and proves Theorem
\ref{thm:atomic}. Finally, in Section \ref{sec:atomicPoA}, we analyze PoA-Robustness for both affine (Corollary \ref{cor:atomicAffinePoA}) and polynomial latency functions (Theorem \ref{thm:atomicPoA}). We also include a threshold-sensitive robustness refinement in Section~\ref{sec:threshold-sensitive-atomic}.

\subsection{The {\sc Potential} Mechanism}
\label{sec:potentialAtomic}

In this section we introduce our main learning-augmented mechanism for
atomic congestion games, namely the \textsc{Potential} mechanism,
along with the technical lemmas used to bound its performance. Analogous to the non-atomic setting, we begin by presenting a base coordination
mechanism, denoted \textsc{Threshold}, which modifies each
latency function based on a specified \say{threshold} congestion level $x_e(s)$ and a corresponding
set of target latencies $h_e$. In particular, for each resource $e$, the mechanism increases the latency  $\ell_e(x)$ to $h_e(x)$ once the congestion level exceeds the threshold. Lemmas
\ref{lemma:enforceAtomic-unique} and \ref{lemma:LiftingAtomic} establish the key structural properties of this mechanism.

We then proceed to the description of the \textsc{Potential}
mechanism. The mechanism takes as input the predicted demand $\bar{r}$
and a parameter $\beta\in \left[{2(d+1)}/{(d+3)},d+1\right]$, and modifies the latency functions according to
the following \textsc{Threshold} mechanism. %

\begin{definition} 
[\textsc{Threshold($h, \x$)}] 
\label{def:thresholdMechAtomic}Given a set of latency functions
$h=(h_e)_{e \in E}$ and a threshold load $\x = (x_e)_{e \in E}$, such
that $\ell_e(x) \le h_e(x)$ for all $x \geq 1$, we define the
following modified latencies:

        $$\hat{\ell}_e(x)=
    \begin{cases}
h_e(x), &  x > x_e\\
\ell_e(x), & \text{otherwise}
    \end{cases}
    $$
\end{definition}

An important property of the \textsc{Threshold} mechanism is that, when the
threshold profile is chosen as a minimizer of the potential $\Phi^h$, this same
profile is also a potential minimizer w.r.t. the modified latencies
$\hat\ell$. Moreover, this particular profile preserves its original social
cost, since it never triggers the increased latencies. We show this property of
\textsc{Threshold} in the following lemma.

\begin{lemma}
\label{lemma:enforceAtomic-unique}
    Given a predicted demand $\bar r$ and a set of non-decreasing latency functions $(h_e)_{e \in E}$, with $\ell_e(x) \le h_e(x)$ for all $x \ge 1$ and $e \in E$, let $\spred$ be a minimizer of $\Phi^h$ for demand $\bar r$, and let $\hat{\ell}$ be the modified latencies produced by \textsc{Threshold}$(h,\spred)$.
    Then $\spred$ is a minimizer of $\Phi^{\hat{\ell}}$ for demand $\bar r$ and preserves the original social cost, i.e., $ \hat C(\spred)=C(\spred)$.
\end{lemma}

\begin{proof}
Let $\hat\Phi=\Phi^{\hat{\ell}}$, $\Phi^h$, and $\Phi^{\ell}$ be the potential under the  latencies $\hat\ell$, $h$, and $\ell$ respectively. Let $\s$ be any strategy profile and let $E_h$ be the set of edges for which the increased latencies $h$ are triggered, i.e., $x_e(\s)>x_e(\spred)$, and let $E_\ell = E\setminus E_h $ be the rest of the edges, where the original latencies are used. The potential of $\s$ is bounded as follows 

\begin{align*}
\hat\Phi({\s}) &\geq \hat\Phi({\s})+\Phi^h(\spred) - \Phi^h(\s)\\
&=\sum_{e\in E_h}\left(\hat\Phi_e({\s})+\Phi_e^h(\spred) - \Phi_e^h(\s)\right) + \sum_{e\in E_\ell}\left(\hat\Phi_e({\s})+\Phi_e^h(\spred) - \Phi_e^h(\s)\right)\\
&=\sum_{e\in E_h}\left(\left(\sum_{i=1}^{x_e(\spred)}\ell_e(i)+\sum_{i=x_e(\spred)+1}^{x_e(\s)}h_e(i)\right)- \sum_{i=x_e(\spred)+1}^{x_e(\s)}h_e(i)\right) + \sum_{e\in E_\ell}\left(\Phi_e^{\ell}({\s})+\Phi_e^h(\spred) - \Phi_e^h(\s)\right)\\
&= \sum_{e\in E_h}\Phi^{\ell}_e(\spred)+\sum_{e\in E_\ell}\left( \Phi^{\ell}_e(\s)-\Phi^h_e(\s)+\Phi^h_e(\spred)\right)\\
&= \sum_{e\in E_h}\Phi^{\ell}_e(\spred)+\sum_{e\in E_\ell}\left(\Phi^{\ell}_e(\s)+ \sum_{i=x_e(\s)+1}^{x_e(\spred)}h_e(i)\right)\\
&= \sum_{e\in E_h}\Phi^{\ell}_e(\spred)+ \sum_{e\in E_\ell}\left( \sum_{i=1}^{x_e(\s)}\ell_e(i)+\sum_{i=x_e(\s)+1}^{x_e(\spred)}h_e(i)\right)\\
&\geq\sum_{e\in E_h}\Phi^{\ell}_e(\spred)+ \sum_{e\in E_\ell} \Phi^{\ell}_e(\spred)\\
&= \hat\Phi({\spred}), 
\end{align*}
where the first inequality is due to the fact that $\spred$ minimizes $\Phi^h$, and the second due to $h_e(x)\geq \ell_e(x)$ for all $e\in E$ and $x\geq 0$. The last equality is due to the fact that the profile $\spred$ never triggers the increased latencies $h$, hence $\hat\Phi(\spred)= \Phi^\ell(\spred)$.

The inequality above shows that $    \hat\Phi(\s)\ge \hat\Phi(\spred)$
for every strategy profile $\s$ feasible for demand $\bar r$. Therefore,
$\spred$ is a minimizer of $\hat\Phi=\Phi^{\hat\ell}$ for demand $\bar r$.
Moreover, since the profile $\spred$ never triggers the increased latencies, we
also have
$
    \hat C(\spred)=C(\spred).
$
\end{proof}

The second key property of \textsc{Threshold} is the one that enables us to derive upper bounds on PoA-Robustness using smoothness arguments that are commonly used in PoA proofs.  The following lemma relates the latencies of the {\em modified network} to the original latencies, serving a role analogous to the Lifting Lemma for non-atomic congestion games (Lemma \ref{lem:LikeVariationalIneq}). 
 Its proof is straightforward.

\begin{lemma}[Lifting Lemma]
\label{lemma:LiftingAtomic}
Given a set of non-decreasing latency functions $h = (h_e)_{e \in E}$,
such that $\ell_e(x) \le h_e(x)$, for all $x \ge 1$ and $e \in E$, and
a threshold load $\x = (x_e(s))_{e \in E}$, let $\hat{\ell}$ be the
modified latencies of \textsc{Threshold}$(h, \x)$. For any demand $r$,
let $\s$ be an equilibrium under $\hat{\ell}$ and $s^*$ any other
feasible strategy profile. Then,
    \[\sum_{e \in E} \hat{\ell}_e(x_e(\s)) x_e(\s) \le \sum_{e \in E} h_e(x_e(\s)+1)x_e(\s^*)\]
\end{lemma}
\begin{proof}
Since $\s$ is a Nash equilibrium w.r.t. $\hat{\ell}$, $$\sum_{e \in s_i}\hat{\ell}_e(x_e(\s))\leq \sum_{e \in s^*_i}\hat{\ell}_e(x_e(\s)+1)\leq \sum_{e \in s^*_i}h_e(x_e(\s)+1).$$    By summing up over all players the lemma follows.
\end{proof}

We are now ready to define the \textsc{Potential} mechanism that receives as input
a predicted demand $\bar r$ and a parameter
$\beta\in \left[{2(d+1)}/({d+3}),d+1\right]$, and modifies the latencies by
using a specific instantiation of \textsc{Threshold} depending on $\bar r$ and $\beta$.

For simplicity, we assume, without loss of generality, that each
resource is decomposed into $d+1$ resources, where the latency of each
such resource is one of the monomials of the latency of the original
resource. In the general case, where latencies are polynomials of
degree $d$, all our expressions can be derived by the summation of the
related expressions of the monomials. For each resource $e$ in the
equivalent setting of monomials, we denote by $d_e$ the degree of the
monomial representing the latency on $e$.

\begin{definition} [\textsc{Potential($\beta$,$\bar r$)}]
\label{def:potentialMechAtomic}
Given a parameter $\beta \in \left[ \frac{2(d+1)}{d+3}, d+1 \right]$, we define the polynomial latency function
$h(\beta) = (h_e)_{e \in E}$ as follows. For $d_e \ge 1$,

    \begin{align*}
        h_e(x) &=  \frac{a_{e}\beta(d+1)}{\beta(d+1)+(d_e+1)(d+1-\beta)}(x^{d_e+1}-(x-1)^{d_e+1}) + \\
        &\qquad \frac{a_{e}(d_e+1)(d+1-\beta)}{\beta(d+1)+(d_e+1)(d+1-\beta)}(x^{d_e}-(x-1)^{d_e})\,.
    \end{align*}
For $d_e=0,$ $h_e$ is defined as $h_e(x) = a_{e}.$ Let $\spred$ be a potential minimizer of $\Phi^h$ for the predicted demand $\bar{r}$. Then, we define the \textsc{Potential}$(\beta, \bar{r})$ mechanism as the learning-augmented mechanism that uses the modified latencies of  \textsc{Threshold}$(h(\beta),\spred)$.
\end{definition}

\begin{remark}
We remark that $h_e$ is chosen as being the function satisfying $\Phi^h_e(x) = \sum_{i=1}^x h_e(i)$, for the following  $\Phi^h_{e}(x)$, for all $e \in E$:
\begin{align}
\label{eq:phiatomic}
    \Phi^h_{e}(x) = a_{e}\frac{\beta(d+1)}{\beta(d+1)+(d_e+1)(d+1-\beta)}x^{d_e+1} + a_{e}\frac{(d_e+1)(d+1-\beta)}{\beta(d+1)+(d_e+1)(d+1-\beta)}x^{d_e}\,, 
\end{align}
for $d_e\geq 1$; for $d_e=0$, it is defined as $\Phi^h_{e}(x) = a_{e}x.$
This potential function is used in \cite{siamcomp/ChristodoulouGG19}
to bound the PoS of $\beta$-approximate equilibria; the strategy profile $\s$ that minimizes
$\Phi^h(\s) = \sum_{e \in E}\Phi^h_{e}(x_e(\s))$ is a $\beta$-approximate
equilibrium under $\ell$. The function $h_e$ is simply
$h_e(x) = \Phi^h_e(x) - \Phi^h_e(x-1)$ for $x \ge 1$, and satisfies
$\ell_e(x)\leq h_e(x)$, for all $x\geq 1$, as needed~\cite{siamcomp/ChristodoulouGG19}, therefore,
\textsc{Threshold}$(h(\beta),\x)$ is well defined.

\end{remark}

We analyze the performance of \textsc{Potential} in the following subsections: Section \ref{sec:AffineAtomic} addresses the case of affine latency functions, while Section \ref{sec:PolynomialsAtomic}  extends the analysis to general  polynomial latencies.

\subsection{Warmup - Affine Latencies}
\label{sec:AffineAtomic}
In this section, we provide consistency-robustness guarantees for the \textsc{Potential} mechanism for affine latencies of the form
$\ell_e(x)=a_ex+b_e$, for each resource $e \in E$, where $a_e$ and
$b_e$ are nonnegative coefficients.

\atomicAffineCor*

\begin{proof}
For a given confidence parameter $\beta\in [1,2]$, $h_e$ is defined as follows
\begin{equation*}
    h_e(x)=\frac{2a_e\beta(2x-1)}{2\beta+2(2-\beta)}+\frac{2a_e(2-\beta)}{2\beta+2(2-\beta)}+b_e=\beta(a_ex+b_e)-(\beta-1)(a_e+b_e)=\beta a_ex +b_e-(\beta-1)a_e\,.\label{eq:hAffineAtomic}
\end{equation*} 

Let  $\spred$ be a potential minimizer of  $\Phi^h$  when the demand is $\bar r$.
The modified latency functions $\hat\ell$ of \textsc{Potential}$(\beta,\bar{r})$  are defined by 
  \textsc{Threshold}$(h, \x)$ , and $\x=(x_e)_{e\in E}$ with $x_e=x_e(\spred)$.

  \textbf{Consistency.}
We first show the consistency guarantee. For that, suppose that the
prediction is accurate, i.e., $r = \bar{r}$. By
Lemma~\ref{lemma:enforceAtomic-unique}, $\spred$ is a minimizer of
$\hat{\Phi}$ and preserves its original social cost, i.e.,
$\hat{C}(\spred)=C(\spred)$. Moreover, by definition, $\spred$ minimizes
$\Phi^h$, and hence also $\Phi^h/\beta$. By
\cite{algorithmica/ChristodoulouKS11}, any minimizer of $\Phi^h/\beta$%
\footnote{In fact, $\Phi^h/\beta$ is used in
\cite{algorithmica/ChristodoulouKS11} to upper bound the PoS of
$\beta$-approximate equilibria.}
has social cost w.r.t. the original latencies $\ell$ at most a factor of
$\frac{ \sqrt{3} + 1}{\sqrt{3} + \beta-1}$ away from the optimal cost.
Since $\hat{C}(\spred)=C(\spred)$, the consistency bound follows.

\textbf{Robustness.}
Let $\s$ be any strategy profile and let $E_h$ be the set of edges for which $x_e(\s)>x_e(\spred)$, i.e., $\hat{\ell}_e(x_e(\s))=h_e(x_e(\s))$, and let $E_\ell \subseteq E\setminus E_h $ be the rest of the edges that are used in $\s$ by at least one agent; note that for each $e\in E_\ell$, it is $\hat{\ell}_e(x_e(\s))=\ell_e(x_e(\s))$. Then $\hat{\Phi}(\s)$ can be written as
\begin{align*}
\hat{\Phi}(\s)
&=
\sum_{e\in E_\ell}
\left(
a_e\frac{x_e(\s)(x_e(\s)+1)}{2}
+
b_ex_e(\s)
\right) \\
&\quad+
\sum_{e\in E_h}
\left(
\beta a_e\frac{x_e(\s)(x_e(\s)+1)}{2}
+
(b_e-(\beta-1)a_e)x_e(\s)
-
(\beta-1)a_e\frac{x_e(\spred)(x_e(\spred)-1)}{2}
\right) \\
&\le
\sum_{e\in E_\ell}
\left(
a_e\frac{x_e(\s)(x_e(\s)+1)}{2}
+
b_ex_e(\s)
\right) \\
&\quad+
\sum_{e\in E_h}
\left(
\beta a_e\frac{x_e(\s)(x_e(\s)+1)}{2}
+
(b_e-(\beta-1)a_e)x_e(\s)
\right) \\
&\le
\sum_{e\in E_\ell}x_e(\s)\ell_e(x_e(\s))
+
\sum_{e\in E_h}x_e(\s)\ell_e(x_e(\s)) \\
&\le
C(\s)\,
\end{align*}

where the second inequality uses $x_e(\s)+1\le2 x_e(\s)$ for $e \in E_\ell$, and for $e \in E_h$. Note that for any $e\in E_h$, we have $a_e (x_e(\s))^2 + b_ex_e(\s) +(\beta/2-1)a_e (x_e(\s))^2 -(\beta/2-1)a_ex_e(\s)\leq x_e(\s)\ell_e(x_e(\s))$, where for the inequality we use that $(\beta/2-1)a_e (x_e(\s))^2\leq (\beta/2-1)a_e x_e(\s)$, since $x_e(\s)\geq 1$ and $\beta\leq 2$. 

Since
\[
\sum_{i=1}^{x_e(\spred)}\ell_e(i)
+
\sum_{i=x_e(\spred)+1}^{x_e(\s)}h_e(i)
\ge
\frac{1}{\beta}\sum_{i=1}^{x_e(\s)}h_e(i)
\ge
\frac{x_e(\s)}{2\beta}h_e(x_e(\s)),
\]
where the first inequality follows from
$\ell_e(i)\ge h_e(i)/\beta$ for the prefix
$i\le x_e(\spred)$, and from $h_e(i)\ge h_e(i)/\beta$ for the suffix
$i>x_e(\spred)$, using $\beta\ge1$ and $h_e(i)\ge0$. The second inequality follows  from
$\sum_{i=1}^{x_e(\s)}h_e(i)\ge \frac{x_e(\s)}{2}h_e(x_e(\s))$, and since also
$\sum_{i=1}^{x_e(\s)}\ell_e(i)\ge \frac{x_e(\s)}{2\beta}\ell_e(x_e(\s))$ for every $e\in E_\ell$, we get
\[
\hat{\Phi}(\s)
\geq
\sum_{e\in E_\ell}\frac{x_e(\s)}{2\beta}\ell_e(x_e(\s))
+
\sum_{e\in E_h}\frac{x_e(\s)}{2\beta}h_e(x_e(\s))
=
\frac{1}{2\beta}\hat{C}(\s)\,.
\]

We combine the above inequalities as follows: Let $\spot$ be any potential minimizer of $\hat{\Phi}$, and $\s^*$ be an optimal profile w.r.t. the original latency functions, i.e., $\s^*$ minimizes $C(\s)$. Then, by utilizing the above inequalities that hold for any $\s$, and therefore for $\spot$ and $\s^*$, it follows that
$\frac{1}{2\beta}\hat{C}(\spot) \leq \hat{\Phi}(\spot)\leq \hat{\Phi}(\s^*) \leq C(\s^*)\,,$
which concludes the proof for robustness, and so completes the proof of the theorem.
\end{proof}

\subsection{Bounded-degree Polynomials}
\label{sec:PolynomialsAtomic}

In this section, we extend our analysis to polynomial latency functions and present the proof of the main theorem. Similarly to the non-atomic case, we relate the proposed equilibrium to a suitable $\beta$-approximate equilibrium to derive the consistency bound. Robustness is established via a potential inequality argument.

\atomicThm*

\begin{proof}
\textbf{Consistency.} We first bound consistency. We obtain that the consistency is equal to the PoS of the $\beta$-approximate equilibria. To that end, suppose that $r = \bar{r}$, i.e., the prediction is correct. By Definition \ref{def:potentialMechAtomic}, the \textsc{Potential}$(\beta,\bar{r})$ mechanism modifies the latencies according to \textsc{Threshold}$(h, \spred)$, where
    \begin{align*}
        h_e(x) &=  \frac{a_{e}\beta(d+1)}{\beta(d+1)+(d_e+1)(d+1-\beta)}(x^{d_e+1}-(x-1)^{d_e+1}) + \\
        &\qquad \frac{a_{e}(d_e+1)(d+1-\beta)}{\beta(d+1)+(d_e+1)(d+1-\beta)}(x^{d_e}-(x-1)^{d_e})\,
    \end{align*}
    $\forall e \in E$ with $d_e\geq 1$, and $h_e(x)=a_e$, $\forall e \in E$, with $d_e=0$, and $\spred$ is a minimizer of $\Phi^h=\sum_{e\in E}\sum_{i=1}^{x_e(\s)} h_e(i)$; recall that
$\Phi^h(\s)=\sum_{e\in E}\Phi^h_e(x_e(\s)),$ where the functions $\Phi^h_e$ are as defined in \eqref{eq:phiatomic}. 
By Lemma~\ref{lemma:enforceAtomic-unique}, $\spred$ is a minimizer of
$\hat{\Phi}=\Phi^{\hat\ell}$ and preserves its original social cost, i.e., $\hat C(\spred)=C(\spred)$. Moreover, by definition, $\spred$ minimizes $\Phi^h$, and hence also
$\Phi^h/\beta$. Then, due to \cite{siamcomp/ChristodoulouGG19} any minimizer of
$\Phi^h/\beta$\footnote{Any minimizer of $\Phi^h/\beta$ corresponds to a
$\beta$-approximate equilibrium.} has social cost w.r.t. the original
latencies $\ell$ at most a factor of $\frac{d+1}{\beta}$ away from the
optimal cost, for
$\beta \in \left[\frac{2(d+1)}{d+3},d+1\right]$. 
Since $\hat C(\spred)=C(\spred)$, we obtain
    $\hat C(\spred)
    \le
    \frac{d+1}{\beta} C_{\mathrm{opt}}(\bar r).$
Therefore, for every
$\beta \in \left[\frac{2(d+1)}{d+3},d+1\right]$, the mechanism's consistency
is at most $\frac{d+1}{\beta}$.

\textbf{Robustness.} 
Let $\s$ be any strategy profile and let $E_h$ be the set of edges for which $x_e(\s)>x_e(\spred)$, i.e. $\hat{\ell}_e(x_e(\s))=h_e(x_e(\s))$ and let $E_\ell = E\setminus E_h $ be the rest of the edges, where the original latencies are used by at least one agent. Then $\hat{\Phi}(\s)$ can be written as:
\[
\hat{\Phi}(\s) = \sum_{e\in E} \sum_{i=1}^{x_e(\s)}\hat{\ell}_e(i) 
=\sum_{e\in E_\ell} \sum_{i=1}^{x_e(\s)}\ell_e(i) 
+ \sum_{e\in E_h} \left(
    \sum_{i=1}^{x_e(\spred)}\ell_e(i)
    +
    \sum_{i=x_e(\spred)+1}^{x_e(\s)}h_e(i)
\right).
\]
We will use the following properties of the auxiliary latencies $h$: for every
edge $e$ and every  $x\ge 1$,
$\ell_e(x)\le h_e(x)\le \beta \ell_e(x)$.\footnote{Regarding the first inequality, for $d_e=0$, we have $h_e(x)=\ell_e(x)=a_e$, so this is immediate. For $d_e \ge 1$, it follows by expanding $h_e(x)-\ell_e(x)$ and using
$x^{d_e}-(x-1)^{d_e}\ge d_e(x-1)^{d_e-1}$; the remaining condition is $\beta(d+1)d_e\ge (d_e+1)(d+1-\beta),$
equivalently $\beta \ge
    \frac{(d_e+1)(d+1)}{d_e(d+1)+d_e+1}$,which holds for all
$\beta\ge \frac{2(d+1)}{d+3}$ for every $1\le d_e\le d$.}
We first show that $\hat{\Phi}(\s)\leq C(\s)$ and then that $\hat{\Phi}(\s)\geq \frac{1}{\beta(d+1)}\hat{C}(\s)$. We will use Lemma~\ref{lem:polyAlgebra}, whose statement and  proof are deferred to Appendix~\ref{appendix:Arithmetic}.

\begin{claim} 
$\hat{\Phi}(\s)\leq C(\s)\,.$
\end{claim} 
\begin{proof}
    For any $e\in E_h$, with $d_e\ge 1$, by using  Lemma~\ref{lem:polyAlgebra} (1), it holds that 
\begin{align*}
& \sum_{i=1}^{x_e(\spred)}\ell_e(i)
+
\sum_{i=x_e(\spred)+1}^{x_e(\s)}h_e(i) \le 
\sum_{i=1}^{x_e(\s)}h_e(i) \\ 
&= \frac{a_{e}\beta(d+1)}{\beta(d+1)+(d_e+1)(d+1-\beta)}\sum_{i=1}^{x_e(\s)}(i^{d_e+1}-(i-1)^{d_e+1}) + \\
        &\qquad \frac{a_{e}(d_e+1)(d+1-\beta)}{\beta(d+1)+(d_e+1)(d+1-\beta)}\sum_{i=1}^{x_e(\s)}(i^{d_e}-(i-1)^{d_e})\\
        &= \frac{a_{e}\beta(d+1)}{\beta(d+1)+(d_e+1)(d+1-\beta)}(x_e(\s))^{d_e+1}+  \frac{a_{e}(d_e+1)(d+1-\beta)}{\beta(d+1)+(d_e+1)(d+1-\beta)}(x_e(\s))^{d_e}\\
        &\leq \frac{a_{e}\beta(d+1)}{\beta(d+1)+(d_e+1)(d+1-\beta)}(x_e(\s))^{d_e+1}+  \frac{a_{e}(d_e+1)(d+1-\beta)}{\beta(d+1)+(d_e+1)(d+1-\beta)}(x_e(\s))^{d_e+1}\\
        &=a_ex_e(\s)^{d_e+1}=x_e(\s)\ell_e(x_e(\s))
\end{align*}
The same also holds with equality when $d_e=0$.  
Then, by using the above inequality for any $e\in E_h$ and since any $\ell_e$ is a non-decreasing function

\begin{align*}
\hat{\Phi}(\s) 
&\leq \sum_{e\in E_\ell} \sum_{i=1}^{x_e(\s)}\ell_e(x_e(\s)) 
+ \sum_{e\in E_h} {x_e(\s)}\ell_e(x_e(\s))\\
&=\sum_{e\in E_\ell} {x_e(\s)}\ell_e(x_e(\s)) 
+ \sum_{e\in E_h} {x_e(\s)}\ell_e(x_e(\s)) 
= C(\s)\,.
\end{align*}
\end{proof}

\begin{claim}
$\hat{\Phi}(\s)\geq \frac{1}{\beta(d+1)}\hat{C}(\s)\,.$
\end{claim} 
\begin{proof} 
By using Lemma~\ref{lem:polyAlgebra} (2) for any $e\in E_h$, with $d_e\geq 1$:
\begin{align*}
\sum_{i=1}^{x_e(\spred)}\ell_e(i)
+
\sum_{i=x_e(\spred)+1}^{x_e(\s)}h_e(i)
&\geq
\frac{1}{\beta}\sum_{i=1}^{x_e(\s)}h_e(i) \\
&=
\frac{1}{\beta}\left(
\frac{a_{e}\beta(d+1)}{\beta(d+1)+(d_e+1)(d+1-\beta)}(x_e(\s))^{d_e+1}\right.\\
&\qquad\left.+
\frac{a_{e}(d_e+1)(d+1-\beta)}{\beta(d+1)+(d_e+1)(d+1-\beta)}(x_e(\s))^{d_e}
\right)\\
&\geq
\frac{1}{\beta}\left(
\frac{a_{e}\beta(d+1)}{\beta(d+1)+(d_e+1)(d+1-\beta)}
x_e(\s)\frac{(x_e(\s))^{d_e+1}-(x_e(\s)-1)^{d_e+1}}{d_e+1}\right.\\
&\qquad\left.+
\frac{a_{e}(d_e+1)(d+1-\beta)}{\beta(d+1)+(d_e+1)(d+1-\beta)}
x_e(\s)\frac{(x_e(\s))^{d_e}-(x_e(\s)-1)^{d_e}}{d_e}
\right)\\
& \geq \frac{x_e(\s)h_e(x_e(\s))}{\beta(d+1)}.
\end{align*}

The first inequality comes from the fact that for the prefix $i\le x_e(\spred)$, we use
$\ell_e(i)\ge h_e(i)/\beta$. For the suffix
$i>x_e(\spred)$, we use $h_e(i)\ge h_e(i)/\beta$, which follows from
$\beta\ge 1$ and $h_e(i)\ge 0$.

The same inequality also holds when $d_e=0$. Moreover, using Lemma~\ref{lem:polyAlgebra} (3) for any $e\in E_{\ell}$:
\[
\sum_{i=1}^{x_e(\s)}\ell_e(i) 
        = \sum_{i=1}^{x_e(\s)}a_e i^{d_e}
        \geq \frac{a_e(x_e(\s))^{d_e+1}}{d_e+1} \geq \frac{x_e(\s)\ell_e(x_e(\s))}{d+1}\,.
\]
Overall, 
\begin{align*}
\hat{\Phi}(\s) 
&=\sum_{e\in E_\ell} \sum_{i=1}^{x_e(\s)}\ell_e(i) 
+ \sum_{e\in E_h} \left(
    \sum_{i=1}^{x_e(\spred)}\ell_e(i)
    +
    \sum_{i=x_e(\spred)+1}^{x_e(\s)}h_e(i)
\right)\\
&\geq \sum_{e\in E_\ell} \frac{x_e(\s)\ell_e(x_e(\s))}{\beta(d+1)} 
+ \sum_{e\in E_h} \frac{x_e(\s)h_e(x_e(\s))}{\beta(d+1)} 
= \frac{1}{\beta(d+1)}\hat{C}(\s)\,.
\end{align*}
\end{proof}

We complete the proof for the robustness upper bound by using the above two claims.
Let $\spot$ be any potential minimizer of $\hat{\Phi}$, and $\s^*$ be an optimal profile w.r.t. the original latency functions, i.e., $\s^*$ minimizes $C(\s)$. Then, by utilizing the above two claims that hold for any $\s$, and therefore for $\spot$ and $\s^*$, it follows that
$\frac{1}{\beta(d+1)}\hat{C}(\spot) \leq \hat{\Phi}(\spot)\leq \hat{\Phi}(\s^*) \leq C(\s^*)\,,$
which concludes the proof for robustness, and so completes the proof of the theorem.
\end{proof}

\subsubsection{A threshold-sensitive robustness refinement}
\label{sec:threshold-sensitive-atomic}

Similarly to the non-atomic case, we present a refined threshold-sensitive analysis of the \textsc{Potential} mechanism, depending on how close the triggered resources are to their predicted thresholds. Let \(\s\) be any strategy profile and let
\(E_h(\s)=\{e\in E:x_e(\s)>x_e(\bar{\s})\}\) be the set of triggered resources. For every \(e\in E_h(\s)\), define
\(\eta_e=x_e(\bar{\s})/x_e(\s)\in[0,1)\). If \(E_h(\s)\neq\emptyset\), let
\(\eta(\s,\bar{\s})=\max_{e\in E_h(\s)}\eta_e\), and otherwise set \(\eta(\s,\bar{\s})=0\). Recall that
\(\hat\Phi_e(\s)=\sum_{i=1}^{x_e(\s)}\hat\ell_e(i)\), \(\hat\Phi(\s)=\sum_{e\in E}\hat\Phi_e(\s)\), and
\(\hat C_e(\s)=x_e(\s)\hat\ell_e(x_e(\s))\).

We will use two elementary facts. First, \(h_e(x)\le \beta\ell_e(x)\) for every resource \(e\) and every integer \(x\ge1\), as established in the proof of Theorem~\ref{thm:atomic}. Second, if
\(\Phi^h_e(x)=\sum_{i=1}^x h_e(i)\), then \(\Phi^h_e(x)\ge xh_e(x)/(d+1)\) for every integer \(x\ge1\). This follows from the definition of the auxiliary potential \(\Phi^h_e\) and Lemma~\ref{lem:polyAlgebra}.

\begin{lemma}[Threshold-sensitive potential lower bound]
\label{lemma:threshold-sensitive-atomic}
For every strategy profile \(\s\),
\[
    \hat{\Phi}(\s)
    \ge
    \frac{1-\left(1-\frac{1}{\beta}\right)\eta(\s,\bar{\s})}
    {d+1}\,
    \hat{C}(\s).
\]
\end{lemma}

\begin{proof}
We prove the inequality resource by resource. First consider an untriggered resource \(e\notin E_h(\s)\), so \(x_e(\s)\le x_e(\bar{\s})\). Then \(\hat\ell_e(i)=\ell_e(i)\) for every \(1\le i\le x_e(\s)\), and hence
\(\hat\Phi_e(\s)=\sum_{i=1}^{x_e(\s)}\ell_e(i)\ge x_e(\s)\ell_e(x_e(\s))/(d+1)=\hat C_e(\s)/(d+1)\), where the inequality follows from Lemma~\ref{lem:polyAlgebra} and the nonnegativity of the coefficients of \(\ell_e\). Since
\(1-(1-1/\beta)\eta(\s,\bar{\s})\le 1\), this already implies the desired resource-wise bound for every untriggered resource.

Now consider a triggered resource \(e\in E_h(\s)\), so \(x_e(\s)>x_e(\bar{\s})\). Then
\begin{align*}
    \hat\Phi_e(\s)
    &=
    \sum_{i=1}^{x_e(\bar{\s})}\ell_e(i)
    +
    \sum_{i=x_e(\bar{\s})+1}^{x_e(\s)}h_e(i)  \\
    &=
    \Phi^\ell_e(x_e(\bar{\s}))+\Phi^h_e(x_e(\s))-\Phi^h_e(x_e(\bar{\s})) \\
    &\ge
    \Phi^h_e(x_e(\s))
    -
    \left(1-\frac{1}{\beta}\right)\Phi^h_e(x_e(\bar{\s})) \\
    &\ge
    \Phi^h_e(x_e(\s))
    -
    \left(1-\frac{1}{\beta}\right)\eta_e\Phi^h_e(x_e(\s)) \\
    &=
    \left(1-\left(1-\frac{1}{\beta}\right)\eta_e\right)\Phi^h_e(x_e(\s)) \\
    &\ge
    \frac{1-\left(1-\frac{1}{\beta}\right)\eta_e}{d+1}x_e(\s)h_e(x_e(\s)) \\
    &=
    \frac{1-\left(1-\frac{1}{\beta}\right)\eta_e}{d+1}\hat C_e(\s).
\end{align*}
The first inequality uses \(h_e(i)\le \beta\ell_e(i)\), which implies
\(\Phi^\ell_e(x_e(\bar{\s}))\ge \Phi^h_e(x_e(\bar{\s}))/\beta\). The second inequality uses
\(\Phi^h_e(x_e(\bar{\s}))\le \eta_e\Phi^h_e(x_e(\s))\): indeed, \(\Phi^h_e\) has nonnegative coefficients and zero constant term as a polynomial in the load, and
\((x_e(\bar{\s})/x_e(\s))^q\le x_e(\bar{\s})/x_e(\s)=\eta_e\) for every \(q\ge1\). The last inequality uses
\(\Phi^h_e(x_e(\s))\ge x_e(\s)h_e(x_e(\s))/(d+1)\). Since
\(\eta_e\le \eta(\s,\bar{\s})\) and \(1-(1-1/\beta)\eta\) is decreasing in \(\eta\), this gives
\[
    \hat\Phi_e(\s)
    \ge
    \frac{1-\left(1-\frac{1}{\beta}\right)\eta(\s,\bar{\s})}
    {d+1}\hat C_e(\s)
\]
for every triggered resource. Summing over all resources completes the proof.
\end{proof}

\begin{corollary}[Threshold-sensitive robustness]
\label{cor:threshold-sensitive-robustness-atomic}
Let \(\spot\) be any potential minimizer of the modified game, and let
\(\s^\ast\) be an optimal strategy profile under the original latencies. Then
\[
    \hat C(\spot)
    \le
    \frac{d+1}
    {1-\left(1-\frac{1}{\beta}\right)\eta(\spot,\bar{\s})}
    C(\s^\ast).
\]
\end{corollary}

\begin{proof}
By Lemma~\ref{lemma:threshold-sensitive-atomic},
\[
    \frac{1-\left(1-\frac{1}{\beta}\right)\eta(\spot,\bar{\s})}
    {d+1}
    \hat C(\spot)
    \le
    \hat\Phi(\spot).
\]
Since \(\spot\) minimizes \(\hat\Phi\), and since the robustness proof of
Theorem~\ref{thm:atomic} shows that \(\hat\Phi(\s)\le C(\s)\) for every strategy profile \(\s\), we have
\(\hat\Phi(\spot)\le \hat\Phi(\s^\ast)\le C(\s^\ast)\). Combining the two inequalities and rearranging gives the claim.
\end{proof}

As in the non-atomic case, this bound refines the worst-case robustness guarantee of Theorem~\ref{thm:atomic} by tracking how close the triggered resources are to their predicted threshold loads. When the realized potential minimizer exceeds the threshold loads by a large margin, the factor is close to \(d+1\); in the worst case, as \(\eta(s^\diamond,\bar s)\) approaches \(1\), the bound recovers the general guarantee \(\beta(d+1)\).

\subsection{PoA-Robustness}\label{sec:atomicPoA}
In this section, we analyze the PoA-Robustness of \textsc{Potential} for the case of atomic games with polynomial latencies with nonnegative coefficients of degree $d$. The main result shows an upper bound of PoA-Robustness as a function of the confidence parameter $\beta$ (Theorem~\ref{thm:atomicPoA}). 

The next result presents an upper bound of PoA-Robustness for the special case of affine latency functions. We stress that this is an improved bound compared to the bound obtained by Theorem~\ref{thm:atomicPoA} for $d=1$. 

\atomicAffinePoACor*

\begin{proof}
Given an arbitrary input rate $r$, let $\s$ be any equilibrium profile under the modified latencies $\hat{\ell}$, and $\s^*$ be the optimal profile under the original latencies $\ell$.

By using the Lifting Lemma for atomic games (Lemma \ref{lemma:LiftingAtomic}) and substituting for $h_e(x) = \beta(a_ex+b_e)-(\beta-1)(a_e+b_e)$, we obtain
\begin{align}
\label{eq:relationToCKS11Atomic}
  \hat{C}(\s) & = \sum_{e \in E}\hat{\ell}_e(x_e(\s))x_e(\s) \le \sum_{e \in E} (\beta a_ex_e(\s) + a_e + b_e)x_e(\s^*).
\end{align}

A quick ---but not tight---  upper bound of PoA-Robustness can be obtained as follows. 
The last term of (\ref{eq:relationToCKS11Atomic}) can be bounded above by  $\beta \cdot \sum_{e \in E}\ell_e(x_e(\s)+1)x_e(\s^*)$, which gives $$\hat{C}(\s)\le \beta \sum_{e \in E}\ell_e(x_e(\s)+1)x_e(\s^*).$$

By observing that the above inequality is also satisfied by $\beta$-approximate equilibria w.r.t. the original latencies, one can follow a similar analysis with that of \cite{algorithmica/ChristodoulouKS11} to upper bound the PoA of $\beta$-approximate equilibria. This gives a PoA-Robustness bound of at most
$\beta\frac{z^2+3z+1}{2z-\beta+1}$, where $z = \left\lfloor \frac{\beta+\sqrt{5+6(\beta-1)+(\beta-1)^2}}{2} \right\rfloor$.

We are able to obtain an improved bound by a more careful analysis (see Figure \ref{fig:robustness_affine} (b)). In particular, by applying Lemma \ref{lem:ineqAtomic} (we defer the proof to Appendix~\ref{appendix:Arithmetic}) to the last term of \eqref{eq:relationToCKS11Atomic} we establish the following claim.

\begin{claim}
    For any parameter $\beta \in [1,2]$ it holds that 
    \begin{align*}
        \hat{C}(\s) \le \frac{\beta}{2z+1}\hat{C}(\s)+\beta \frac{z^2+(1+\frac{2}{\beta})z+\frac{\beta+3}{4\beta}}{2z+1}C(\s^*) \,. 
    \end{align*}
\end{claim}

\begin{proof}
\begin{align*}
    \hat{C}(\s) & \leq \sum_{e \in E} (\beta a_ex_e(\s) + a_e + b_e)x_e(\s^*) \\ & \le \sum_{e \in E} \beta\left(\frac{1}{2z+1}\right)a_e(x_e(\s))^2 + \beta\frac{z^2+\left(1+\frac{2}{\beta}\right)z+\frac{\beta+3}{4\beta}}{2z+1}a_e(x_e(\s^*))^2 + b_e x_e(\s^*) \\
    & \le \frac{\beta}{2z+1} \hat{C}(\s) + \frac{\beta z^2+(\beta+2)z+\frac{\beta+3}{4}}{2z+1}C(\s^*)\,,
\end{align*}
where the last inequality follows from the fact that $\beta z^2+(\beta+2)z+\frac{\beta+3}{4} \geq 2z+1$ for all non-negative integers $z$ and $\beta\geq 1.$ 
\end{proof}

Rearranging yields the following bound:

$$\frac{\hat{C}(\s)}{C(\s^*)} \le \frac{\beta z^2+(\beta+2)z+\frac{\beta+3}{4}}{2z+1-\beta} \,.$$

Taking \(z\) to be the integer specified in the statement gives the claimed
upper bound.

\end{proof}

Now we are ready to restate the main result of this section.

\atomicPoAThm*

\begin{proof}
In fact we will show that $$\mbox{PoA-Robustness: }\frac{\beta((z+1)^{2d+1}-z^{d+1}(z+2)^d)}{(z+1)^{d+1} - z^{d+1} - \beta ((z+2)^d-(z+1)^d)}, $$
    where $z$ is the maximum integer that satisfies $\frac{z^{d+1}}{(z+1)^d}\leq \beta.$ Then, similarly to \cite{siamcomp/AlandDGMS11} it holds that the above expression is upper bounded by $z^{d+1}=\phi_{d,\beta}^{d+1}$, for $z$ that makes $\frac{z^{d+1}}{(z+1)^d} = \beta$.   

Given any input rate $r$, let $\s$ be an equilibrium profile under the modified latencies $\hat{\ell}$ and $\s^*$ the optimal profile under the original latencies $\ell$. Using the Lifting Lemma for atomic games (Lemma \ref{lemma:LiftingAtomic}), we obtain that:

\begin{align}
    \hat{C}(\s)=\sum_{e \in E}\hat{\ell}_e(x_e(s))x_e(s) \le \sum_{e \in E}h_e(x_e(s)+1)x_e(s^*) 
\end{align}

Now, the right part can be upper bounded by $\beta \sum_{e \in E}\ell_e(x_e(s)+1)x_e(s^*).$
The resulting inequality enables us to apply the analysis of the atomic PoA for $\beta$-approximate equilibria in \cite{siamcomp/ChristodoulouGG19}, by observing that $\s$ corresponds to a $\beta$-approximate equilibrium under the original latencies $\ell$. Thus, we can upper bound the robustness via the PoA of $\beta$-approximate equilibria provided by \cite{siamcomp/ChristodoulouGG19}. Hence, we obtain a robustness bound of at most $\frac{\beta((z+1)^{2d+1}-z^{d+1}(z+2)^d)}{(z+1)^{d+1} - z^{d+1} - \beta ((z+2)^d-(z+1)^d)},$
    where $z$ is the maximum integer that satisfies $\frac{z^{d+1}}{(z+1)^d}\leq \beta.$
\end{proof}

\section{Impossibility Results}
\label{sec:impossibilityResults}

In this section, we present impossibility results concerning the design of learning-augmented coordination mechanisms that satisfy further desirable properties.
In Section~\ref{sec:poaCons} we show that when a stronger notion of consistency (PoA-Consistency) is required, any mechanism that is $1$-consistent with respect to this stronger criterion must suffer unbounded robustness. Furthermore, in Section~\ref{sec:smoothness} we establish impossibility results for mechanisms whose approximation preserves a continuity property for arbitrarily small prediction errors.

\subsection{PoA-Consistency}
\label{sec:poaCons}

The following two theorems demonstrate that designing mechanisms that are 1-consistent with respect to the PoA is excessively demanding, thereby, motivating our choice of the potential minimizers for the consistency and robustness guarantees. In particular, we show that mechanisms that are $1$-consistent w.r.t.  PoA exhibit unbounded robustness (Theorem~\ref{thm:impossibilityPoAconsNon-atomic} and Theorem~\ref{thm:impossibilityPoAcons} for non-atomic and atomic congestion games, respectively). Furthermore, for the atomic games we generalize our impossibility result for mechanisms with PoA-Consistency equal to $\alpha\geq 1$ (see Appendix~\ref{sec:impossibilitiesPoAadditional}).

\begin{theorem}
\label{thm:impossibilityPoAconsNon-atomic}
    There is no learning-augmented coordination mechanism for non-atomic congestion games that has simultaneously PoA-Consistency equal to $1$ and finite robustness.
\end{theorem}

\begin{proof}    
    Consider a non-atomic congestion game defined as follows. There are $4$ commodities, denoted by $0, 1, 2,3$, and a set of resources $E=\{e_0,e_1,e_2,e_3\}$. The predicted demand is $\bar{r}_0=2$, $\bar{r}_1=\bar{r}_2=1$, $\bar r_3=0$. The latencies of the resources are: $\ell_0(x)=2$, $\ell_1(x)=\ell_2(x)=x$, and $\ell_3(x)=\epsilon$. The set of strategies is $\Sigma_0=\{\{e_0\},\{e_1\},\{e_2\}\}$ for commodity $0$, $\Sigma_1=\{\{e_1\},\{e_2\}\}$ for commodity $1$,  $\Sigma_2=\{\{e_1,e_3\},\{e_2\}\}$ for commodity $2$, and $\Sigma_3=\{e_3\}$. 

    The optimal solution assigns $\bar r_0$ to $e_0$, $\bar r_1$ to $e_1$, and $\bar r_2$ to $e_2$ for the predicted demand $\bar r$ resulting in a social cost $C_{opt}(\bar r)=2\ell_0(2)+\ell_1(1)+\ell_2(1)=6$. There exists an equilibrium profile $f$ in which commodity $0$ routes half of its demand through $e_1$ and the rest through $e_2$, commodity $1$ routes all its demand through $e_1$, while commodity $2$ routes all its demand through $\{e_2\}$, yielding a social cost $C(f)=8$.

    Now consider any mechanism that is required to be $1$-consistent w.r.t. PoA. So, when $r=\bar r$, the mechanism must enforce commodity $0$ to route all its demand through $e_0$, and for this it must modify the latencies of $e_1$ and $e_2$, so that $\hat{\ell}_1(x) > 2$ for $x > 1$ and $\hat{\ell}_2(x) > 2$ for $x > 1$. Then,
    
    the profile where commodity $0$ routes all its demand through $e_0$, commodity $1$ routes all its demand through $e_2$ and commodity $2$ routes all its demand through $\{e_1,e_3\}$ must not be an equilibrium for $r=\bar r$. Eliminating such an equilibrium without harming the predicted optimal solution requires $\hat{\ell}_3(x) > 1$ for $x > 0$.

    Now consider an actual instance in which only commodity $3$'s demand is realized, i.e., $r_0=r_1=r_2=0$ and $r_3=1$. In this case, the social cost of the unique outcome is equal to the modified latency of $e_3$ for the realized demand. Hence, the mechanism's robustness is at least ${1}/{\epsilon}$, which tends to $\infty$ as $\epsilon \to 0$.
\end{proof}

\begin{theorem}
\label{thm:impossibilityPoAcons}
    There is no learning-augmented coordination mechanism for atomic congestion games that has simultaneously PoA-Consistency equal to $1$ and finite robustness.
\end{theorem}

\begin{proof}
    Consider an atomic congestion game defined as follows. The set of players is $N=\{1,2,3\}$ and the set of resources is given by $E=\{e_1,e_2,e_3\}$. The strategy sets are $\Sigma_1=\{\{e_1\},\{e_2\}\}$, $\Sigma_2=\{\{e_1,e_3\},\{e_2\}\}$, and $\Sigma_3=\{e_3\}$, for players $1, 2,$ and $3$ respectively.  
    The latency functions are defined as $\ell_1(x)=\ell_2(x)=x$ and $\ell_3(x)=\epsilon$, for $\epsilon>0$ arbitrarily small.

    Let the predicted demand vector be $\bar{r}=\{1,2\}$, i.e., indicating that only players $1$ and $2$ will be active. For this predicted instance, the socially optimal allocation assigns agent $1$ to $e_1$ and agent $2$ to $e_2$, resulting in a social cost of $C_{opt}(\bar r)=\ell_1(1)+\ell_2(1)=2$.
    However, for the same predicted active set $\{1,2\}$, the strategy profile $\s$ in which player $1$ chooses $\{e_2\}$ and player $2$ chooses $\{e_1,e_3\}$ is a Nash equilibrium with a resulting social cost of $C(s)=\ell_1(1)+\ell_2(1)+\ell_3(1)=2+\epsilon$, so the PoA for the predicted instance is  equal to $1+\epsilon/{2}>1$. 

    Now consider any mechanism that is required to be $1$-consistent w.r.t. the PoA: the mechanism must modify the latencies such that any equilibrium has optimal social cost. Concretely, the mechanism must eliminate the bad equilibrium $\s$. Observe that the resources $e_1$ and $e_2$ are also used in the optimal solution for the predicted instance, so increasing their latencies for load 1 would also yield an increased social cost of the optimal allocation. Therefore, the only way to eliminate the bad equilibrium without harming the predicted optimum (i.e., maintaining $1$-consistency) is to enforce player 2 to deviate to $\{e_2\}$ by satisfying $\hat{\ell}_3(1)\ge c$, for some constant $c > 1$.
    
    Now consider an actual realization, that differs from the predicted demand, in which only player $3$ is active. In this case, the social cost of the unique outcome is the latency of $e_3$, under the mechanism's modified latencies that is $\hat{\ell}_3(1) > c$. Hence, robustness is equal to ${c}/{\epsilon}$, which is arbitrarily large after choosing $\epsilon$ arbitrarily small. 
\end{proof}

These impossibilities show that the PoA is a poorly suited consistency target for mechanisms that are required to preserve optimality when the prediction is correct. 

\subsubsection{PoA-Consistency lower bounds}
\label{sec:impossibilitiesPoAadditional}

We further show lower bounds relaxing the $1$-PoA-consistency requirement. We show that any $\alpha$-consistent mechanism must suffer the following PoA-Robustness.

\begin{theorem}
\label{thm:impossibilityPoAconsTradeoff}
    Any $\alpha$-consistent mechanism (with respect to PoA) must have robustness 

    $$\rho \ge 
    \begin{cases}
        (3-2\alpha)/(2\alpha-2), & 1 \le \alpha \le (1+\sqrt2)/2\\
        2/(2\alpha-1), & (1+\sqrt2)/2\le\alpha<(1+\sqrt17)/4
    \end{cases}
    $$
\end{theorem}

\begin{proof}

    Consider the same construction as in Theorem~\ref{thm:impossibilityPoAcons}, but let the ground
    set of players be \(N=\{1,2,3,4\}\), the resources be \(M=\{e_1,e_2,e_3\}\), and the strategy sets
    be as before for players \(1,2,3\), with player \(4\) having the singleton strategy set
    \(\Sigma_4=\{\{e_1\}\}\). The predicted active set is \(\bar N=\{1,2\}\). The latency functions are
    \(\ell_1(x)=\ell_2(x)=x\) and \(\ell_3(x)=2\alpha-2+\epsilon\), where \(\alpha\ge1\) denotes the
    consistency guarantee and \(\epsilon>0\) is arbitrarily small and chosen such that
    \(2\alpha-2+\epsilon<1\).

    For the predicted instance, the socially optimal allocation assigns player \(1\) to \(e_1\) and
    player \(2\) to \(e_2\), and has social cost \(2\). As in Theorem~\ref{thm:impossibilityPoAcons},
    the profile \(s\) in which player \(1\) selects \(e_2\) and player \(2\) selects \(\{e_1,e_3\}\)
    is an equilibrium under the original latencies. Its social cost is $
        2+(2\alpha-2+\epsilon)=2\alpha+\epsilon,$
    and therefore its approximation ratio is \(\alpha+\epsilon/2>\alpha\). Thus, any
    \(\alpha\)-consistent mechanism with respect to PoA must prevent \(s\) from being an equilibrium
    under the modified latencies.

    In the profile \(s\), player \(1\)'s cost is \(\hat{\ell}_2(1)\), while deviating to \(e_1\) would
    create load \(2\) on \(e_1\), and hence would cost at least \(2\), since the mechanism only
    increases latencies. Similarly, player \(2\)'s cost is \(\hat{\ell}_1(1)+\hat{\ell}_3(1)\), while
    deviating to \(e_2\) would create load \(2\) on \(e_2\), and hence would cost at least \(2\).
    Therefore, in order to eliminate \(s\), the mechanism must satisfy
    \[
        \hat{\ell}_2(1)\ge2
        \qquad\text{or}\qquad
        \hat{\ell}_1(1)+\hat{\ell}_3(1)\ge2.
    \]

    We first argue that the case \(\hat{\ell}_2(1)\ge2\) is incompatible with
    \(\alpha\)-PoA-consistency in the range of the theorem. Indeed, since
    \(\alpha<(1+\sqrt{17})/4<3/2\), we have \(2\alpha<3\). If \(\hat{\ell}_2(1)\ge2\), then the
    predicted optimal profile has modified social cost at least $\hat{\ell}_1(1)+\hat{\ell}_2(1)\ge 1+2=3>2\alpha$,
    because \(\hat{\ell}_1(1)\ge\ell_1(1)=1\). The non-optimal profile \(s\) has modified social cost
    at least $\hat{\ell}_1(1)+\hat{\ell}_2(1)+\hat{\ell}_3(1)
        \ge 1+2+(2\alpha-2+\epsilon)>2\alpha$.
    The remaining two profiles place both active players on the same load-dependent resource, and
    therefore have modified social cost at least \(4>2\alpha\). Since every finite congestion game has
    a pure equilibrium, some profile must be an equilibrium, and all possible profiles have modified
    social cost greater than \(2\alpha\), contradicting \(\alpha\)-PoA-consistency. Hence, for an
    \(\alpha\)-consistent mechanism, we may assume $\hat{\ell}_1(1)+\hat{\ell}_3(1)\ge2.$ Under this condition, the non-optimal profile \(s\) has modified social cost at least $\hat{\ell}_1(1)+\hat{\ell}_2(1)+\hat{\ell}_3(1)\ge 1+2=3>2\alpha.$
    The two profiles in which both active players create load \(2\) on either \(e_1\) or \(e_2\) also
    have modified social cost at least \(4>2\alpha\). Thus, in any \(\alpha\)-consistent modified
    game, the only possible equilibrium for the predicted active set is the predicted optimal profile.
    Since this equilibrium must have modified social cost at most \(2\alpha\), we get $\hat{\ell}_1(1)+\hat{\ell}_2(1)\le2\alpha$.     Using \(\hat{\ell}_2(1)\ge1\), this implies
    $\hat{\ell}_1(1)\le2\alpha-1$. Combining this with \(\hat{\ell}_1(1)+\hat{\ell}_3(1)\ge2\), we obtain $
        \hat{\ell}_3(1)\ge2-\hat{\ell}_1(1)$.

     We now derive the robustness lower bound. Consider an actual realization in which only player
    \(4\) is active. Since player \(4\) is forced to use \(e_1\), the unique equilibrium has modified
    social cost \(\hat{\ell}_1(1)\), while the optimal social cost under the original latencies is \(1\).
    Hence, $
        \rho\ge \hat{\ell}_1(1)$.
    Next consider an actual realization in which only player \(3\) is active. Since player \(3\) is
    forced to use \(e_3\), the unique equilibrium has modified social cost \(\hat{\ell}_3(1)\), while
    the optimal social cost under the original latencies is \(2\alpha-2+\epsilon\). Hence, $
        \rho\ge \frac{\hat{\ell}_3(1)}{2\alpha-2+\epsilon}$.
    Therefore,
    \[
        \rho
        \ge
        \max\left\{
            \hat{\ell}_1(1),
            \frac{\hat{\ell}_3(1)}{2\alpha-2+\epsilon}
        \right\}
        \ge
        \max\left\{
            \hat{\ell}_1(1),
            \frac{2-\hat{\ell}_1(1)}{2\alpha-2+\epsilon}
        \right\},
    \]
    with \(\hat{\ell}_1(1)\le2\alpha-1\). Letting \(\epsilon\to0\), the right-hand side is minimized by balancing the two terms, namely when
    $\hat{\ell}_1(1)=\frac{2-\hat{\ell}_1(1)}{2\alpha-2}$,
    which gives $\hat{\ell}_1(1)=\frac{2}{2\alpha-1}.$ This balancing point is feasible precisely when $
        \frac{2}{2\alpha-1}\le2\alpha-1,$
    equivalently, when \(\alpha\ge(1+\sqrt2)/2\). Therefore, if
    \((1+\sqrt2)/2\le\alpha\), then $
        \rho\ge \frac{2}{2\alpha-1}$.
    If \(1\le\alpha\le(1+\sqrt2)/2\), then the balancing point is infeasible, and the best lower
    bound is obtained at the endpoint \(\hat{\ell}_1(1)=2\alpha-1\), yielding
    \[
        \rho
        \ge
        \frac{2-(2\alpha-1)}{2\alpha-2}
        =
        \frac{3-2\alpha}{2\alpha-2}.
    \]
    For \(\alpha=1\), this expression is interpreted as an unbounded lower bound by taking
    \(\epsilon\to0\).

    Finally, the second lower bound is relevant as long as it is at least the consistency parameter
    itself, i.e.,
    $\frac{2}{2\alpha-1}\ge\alpha.$
    Solving this inequality gives \(\alpha\le(1+\sqrt{17})/4\). Hence, for
    \((1+\sqrt2)/2\le\alpha<(1+\sqrt{17})/4\), we obtain
    \[
        \rho\ge \frac{2}{2\alpha-1}.
    \]
    This completes the proof.
\end{proof}

\subsection{Error Tolerance}
\label{sec:smoothness}

In this section, we examine whether it is possible to design learning-augmented coordination mechanisms  whose performance degrades gracefully as the prediction error increases in non-atomic games. Inspired by the notion of  \emph{brittleness} in learning-augmented algorithms \cite{nips/Elenter0DL24:britleness} (see also \cite{aistats/BenomarP25}), we formalize this notion of smooth performance degradation by requiring that the performance measure varies continuously with the prediction error, for sufficiently small errors. For our lower bounds, we use the standard monotonicity assumption, also used in \cite{algorithmica/ChristodoulouMP14}, that both the original and modified latencies are nondecreasing.

Our main impossibility (Theorem~\ref{thm:smoothNonSymmetric}) shows that no mechanism can be $1$-consistent while simultaneously achieving a graceful degradation to $4/3$, which coincides with the PoA without predictions. In other words, in order to improve the PoA of the original network, any mechanism achieving optimal consistency suffers a discontinuous performance that collapses to the original PoA value, even for arbitrarily small prediction error.  This phenomenon highlights fundamental barriers to designing $1$-consistent mechanisms in this setting.

We complement this negative result with a refined analysis that yields both positive and negative results: we show that this impossibility persists (albeit with weaker bounds) even in symmetric games (Theorem~\ref{thm:smoothImpossibility}) and (multi-commodity) network congestion games (Theorem~\ref{thm:SmoothnessMC}). Nonetheless, in the special case of single-commodity network congestion games, we exploit the network structure to obtain a positive result (Theorem~\ref{thm:SmoothnessSC}).

\paragraph{Congestion Games.} We consider the most general setting: non-atomic congestion games with no underlying network structure. We establish the following impossibility result.

\begin{theorem}\label{thm:smoothNonSymmetric}
There is no $1$-consistent learning-augmented coordination mechanism for non-atomic congestion games that guarantees the existence of an equilibrium with approximation ratio strictly better than $4/3$ as the prediction error tends to zero.
\end{theorem}

\begin{proof}
    Consider a non-atomic congestion game defined as follows: there are two commodities $1$ and $2$, and a set of resources $E=\{e_1, e_2, e_3, e_4\}$. The set of latencies is given below: $\ell_1(x)=\ell_4(x)=1$ and $\ell_2(x)=\ell_3(x)=x$. The set of strategies is $\Sigma_1=\{\{e_1,e_2\},\{e_2,e_3\}\} = \{S_1,S_2\}$ for commodity $1$, and $\Sigma_2=\{\{e_2,e_3\},\{e_3,e_4\}\}=\{S_2,S_3\}$. For each commodity $i$ the predicted demand is  $\bar r_i=1/2$. The optimal solution assigns $\bar r_1$ to $S_1$ and $\bar r_2$ to $S_3$, with a social cost $C_{opt}=1/2(\ell_1(1/2)+\ell_2(1/2))+1/2(\ell_3(1/2)+\ell_4(1/2)) = 3/2$. There exists a Nash equilibrium in which both commodities route all their demands through $S_2$; its social cost is $C_N=1/2(\ell_2(1)+\ell_3(1))+1/2(\ell_2(1)+\ell_3(1))=2$. Hence, the PoA is equal to $4/3$.

    Consider any $1$-consistent coordination mechanism. Suppose that $\hat\ell_2(x) < 1$ for some $x>1/2$; then commodity $2$ has a profitable deviation from the optimal allocation: commodity $2$ improves its latency by routing $x-1/2$ via $S_2$ and the rest via $S_3$ compared to routing all its demand via $S_3$. Observe that this also holds due to symmetry for $\hat\ell_3(x)$ and commodity $1$. Combining the two facts, we obtain that $\hat{\ell}_2(x) \geq 1$ and $\hat{\ell}_3(x) \geq 1$ for all $x > 1/2$. 

    Now let the actual demands be $r_1=r_2=(1+\epsilon)/{2}$ for $\epsilon>0$ infinitesimally small. Note that $e_2$ belongs to all strategies of commodity $1$ and $e_3$ belongs to all strategies of commodity $2$; therefore, no matter the strategy profile, the load of $e_2$ and $e_3$ is more than $1/2$, resulting in their latency being at least $1$. Hence, the social cost of any equilibrium in the modified game is at least $2(1+\epsilon)$ which tends to $2$ as $\epsilon$ goes to $0$, while the optimal solution in the original network has social cost that tends to $3/2$ as $\epsilon$ goes to $0$. We conclude that any equilibrium for this almost accurate demand is at least $4/3$. 
\end{proof}

\begin{theorem}
\label{thm:smoothImpossibility}
    There is no $1$-consistent coordination mechanism in non-atomic symmetric congestion games that guarantees the existence of an equilibrium with approximation ratio strictly better than $7/6$ as the prediction error tends to zero.
\end{theorem}

\begin{proof}
    Consider the symmetric congestion game $(G,\ell)$ defined as follows. The predicted demand is $\bar{r}=1$ and the set of resources is given by $E = \{e_1, e_2, e_3, e_4\}$. The set of available strategies is $\Sigma = \{S_{1}, S_{2}, S_{3}\}$, where  $S_{1}=\{e_1,e_3\}$, $S_{2}=\{e_2, e_4\},$ and $S_{3}=\{e_1,e_2\}$, and the latency functions are defined as $\ell_1(x)=\ell_2(x)=x$ and $\ell_3(x)=\ell_4(x)=1$.

    Notice that this setup closely resembles the Braess network. The optimal solution assigns half of the demand to resources in $S_{1}$ and the other half to $S_{2}$, resulting in a total latency of $C_{opt}(\bar{r})=1/2(\ell_1(1/2)+\ell_3(1/2))+1/2(\ell_2(1/2)+\ell_4(1/2))=3/2$. However, in the unique equilibrium $f$, all players choose strategy $S_{3}$, yielding a total latency of $C_N(f)=2$. 
    
    Regarding consistency, let $\bar{f^*}$ denote the optimal solution under the predicted input rate $\bar{r}$, and let $\hat{\ell}_e$ denote the modified latency function for each resource $e \in E$.  Similarly to the characterization of $1$-consistent mechanisms for parallel link networks given by \cite{CCSV25:predictions}, we establish necessary conditions on the modified latencies for $x>\bar{f_e^*}$.

    \begin{claim}
        For any $1$-consistent mechanism with modified latencies $\hat{\ell}$, the following hold for any $e\in E$:

        \begin{itemize}
            \item $\hat{\ell}_e(\bar{f}_e^*)=\ell_e(\bar{f}_e^*)$ and 
            \item $\liminf_{\epsilon \rightarrow0^+} \hat{\ell}_e(\bar{f}_e^*+\epsilon) \ge 1$.
        \end{itemize}
    \end{claim}
    \begin{proof}
        Let $r = \bar r = 1.$ We first establish the first condition. Observe that the optimal solution of the congestion game is unique. Thus, in order to guarantee $1$-consistency, $\bar{f}^*$ should be an equilibrium with respect to $\hat{\ell}$, and the latencies must remain untouched for any used resource in the optimal solution. Since all resources are used in the optimal solution, it should be $\hat{\ell}_e(\bar{f}_e^*)=\ell_e(\bar{f}_e^*)$, for all $e\in E$.

        Regarding the second condition, it should be $\hat{\ell}_3(x) \ge \ell_3(x)= 1$ and $\hat{\ell}_4(x) \ge \ell_4(x)= 1$ for all $x>0$. Regarding the rest of the resources, assume for the sake of contradiction that $\liminf_{\epsilon \rightarrow0^+} \hat{\ell}_1(\bar{f}_1^*+\epsilon) < 1$; the proof for $e_2$ is similar. Then, the optimal solution (that splits half of the rate through $S_{1}$ and half through $S_{2}$) is not an equilibrium. This is because the cost for strategy $S_2$ equals $3/2$, whereas deviating to strategy $S_3$ results in a cost strictly lower than $3/2.$ 
    \end{proof}

    Let now $r=1+\epsilon$, where $\epsilon>0$. Then, in any equilibrium $f$,
either $f_1>1/2$ or $f_2>1/2$ (or both).
If w.l.o.g. $f_1>1/2$ and $f_2\le1/2$, then
$f_{S_2}+f_{S_3}=f_2\le1/2$, leaving
$f_{S_1}=r-f_2>1/2$ units of flow to experience a cost of
at least $2$. The remaining $f_2$ units experience a cost
of at least $1+f_2$.
If both $f_1>1/2$ and $f_2>1/2$, all $r$ units experience
a cost of at least $2$.
Overall, let $x=f_{S_1}$ in the first case and $x=r$ in the
second. Then $x>1/2$ units experience a cost of at least $2$,
and the remaining $r-x$ units experience a cost of at least
$1+r-x$. Therefore,
\[
\hat C(f)
\ge 2x+(r-x)(1+r-x)
=2r-\frac14+\left(r-x-\frac12\right)^2
\ge 2r-\frac14.
\]

In the original game, splitting the demand equally between
$S_1$ and $S_2$ gives social cost $r+r^2/2$.
Hence $C_{opt}(r)\le r+r^2/2$.
Taking the lower limit as $\epsilon$ goes to $0$, for any equilibrium $f$ at demand $1+\epsilon$, yields
\[
\liminf_{\epsilon\downarrow0}
\frac{\hat C(f)}{C_{opt}}
\ge
\lim_{\epsilon\downarrow0}
\frac{2(1+\epsilon)-1/4}
     {(1+\epsilon)+(1+\epsilon)^2/2}
=\frac76.
\]
\end{proof}

\paragraph{Network congestion games: Multi-commodity.} Having established a strong negative result in the general setting, we now examine whether additional structure can facilitate perfect consistency along with error tolerance guarantees. The natural restriction is to consider multi-commodity network congestion games, where players' strategies correspond to paths in a graph. Despite the added structure, we show that the impossibility persists even in this more constrained setting. Specifically, we establish that $1$-consistency cannot be combined
with an approximation ratio that approaches $1$ as the prediction
error tends to zero.

\begin{theorem}
\label{thm:SmoothnessMC}
    There is no $1$-consistent coordination mechanism in multi-commodity network non-atomic congestion games that guarantees the existence of an equilibrium with approximation ratio strictly better than $7/6$ as the prediction error tends to zero.
\end{theorem}

\begin{proof}
Consider a multi-commodity game defined over the Wheatstone network (also known as the Braess paradox network). The network consists of five edges $E=\{e_1, e_2, e_3, e_4, e_5\}$ with latencies $\ell_1(x)=\ell_4(x)=x$, $\ell_2(x)=\ell_3(x)=1$, and 
$\ell_5(x)=0$, and two commodities, $(s_1, t_1)$ and $(s_2, t_2)$ (as illustrated in Figure~\ref{fig:NoSmoothNetMulti}). Commodity $2$ has a unique strategy: all flow from $s_2$ to $t_2$ is routed along edge $e_5$.

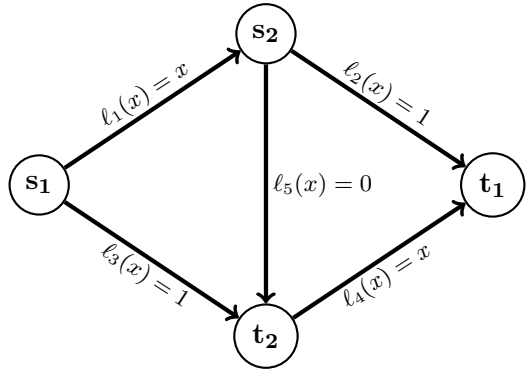
\begin{figure}[h]
\centering
\begin{tikzpicture}[xscale=1,yscale=1]

\node[draw,circle,thick] at (0,0) (s1)  {{$\mathbf{s_1}$}} ;
\node[draw,circle,thick] at (6,0) (t1)  {{$\mathbf{t_1}$}} ;
\node[draw,circle,thick] at (3,2) (s2)  {{$\mathbf{s_2}$}} ;
\node[draw,circle,thick] at (3,-2) (t2)  {{$\mathbf{t_2}$}} ;

\draw[ultra thick,->] (s1) to (s2) ;
\node[rotate=33] at (1.4,1.2) {\footnotesize $\ell_1(x)=x$};
\draw[ultra thick,->] (s2) to (t1) ;
\node[rotate=-33] at (4.6,1.2) {\footnotesize $\ell_2(x)=1$};
\draw[ultra thick,->] (s1) to (t2) ;
\node[rotate=-33] at (1.4,-1.2) {\footnotesize $\ell_3(x)=1$};
\draw[ultra thick,->] (t2) to (t1) ;
\node[rotate=33] at (4.6,-1.2) {\footnotesize $\ell_4(x)=x$};
\draw[ultra thick,->] (s2) to (t2) ;
\node at (3.75,0) {\footnotesize $\ell_5(x)=0$};

\end{tikzpicture}
\caption{Network multicommodity game.}
\label{fig:NoSmoothNetMulti}
\end{figure}

Let $\bar{r}_1=\bar{r}_2=1$. In the optimal flow $\bar {r}_1$ uses only paths $(s_1,s_2,t_1)$ and $(s_1,t_2,t_1)$, whereas in the equilibrium flow it uses only the path $(s_1,s_2,t_2,t_1)$; $\bar{r}_2$ uses only $e_5$ in both optimum and equilibrium flow. Any 1-consistent mechanism should preserve $\hat{\ell}_5(1)=0$. 

Then, $1$-consistency requires that the predicted optimal flow be an equilibrium under the modified latencies,
with $\hat{\ell}_1(1/2)=\hat{\ell}_4(1/2)=1/2$ and
$\hat{\ell}_2(1/2)=\hat{\ell}_3(1/2)=1$.
Moreover,
$\liminf_{\epsilon\downarrow0}
[\hat{\ell}_1(1/2+\epsilon)+\hat{\ell}_5(1+\epsilon)]\ge1$.
Otherwise, commodity $1$ could profitably deviate from
$(s_1,t_2,t_1)$ to $(s_1,s_2,t_2,t_1)$: edge $e_4$ is common
to both paths, while replacing $e_3$ by $e_1$ and $e_5$
would incur a cost strictly below $1$.
By symmetry, the same inequality holds with $\hat{\ell}_4$
in place of $\hat{\ell}_1$.

Let now $r_1=r_2=1+\epsilon$, where $\epsilon>0$, and let $f$
be any equilibrium. The flow identities give
$f_1+f_4=f_5\ge r_2>1$ and
$f_2+f_3=2r_1-f_1-f_4$.
Thus, either $f_1>1/2$ or $f_4>1/2$.
Assume w.l.o.g. that $f_1>1/2$.
By monotonicity and the preceding inequality,
$\hat{\ell}_1(f_1)+\hat{\ell}_5(f_5)\ge1$.
Using also $\hat{\ell}_4(f_4)\ge f_4$,
$\hat{\ell}_5(f_5)\ge0$, and
$\hat{\ell}_2(f_2),\hat{\ell}_3(f_3)\ge1$, we obtain
\[
\begin{aligned}
\hat C(f)
&\ge
f_1[\hat{\ell}_1(f_1)+\hat{\ell}_5(f_5)]
+f_4[\hat{\ell}_4(f_4)+\hat{\ell}_5(f_5)]
+f_2+f_3\\
&\ge f_1+f_4^2+2r_1-f_1-f_4\\
&=2r_1-\frac14+\left(f_4-\frac12\right)^2
\ge 2r_1-\frac14.
\end{aligned}
\]

In the original game, splitting the demand of commodity $1$
equally between the two outer paths gives social cost
$r_1+r_1^2/2$, while commodity $2$ incurs zero cost.
Hence, $C_{opt}(r_1,r_2)\le r_1+r_1^2/2$.
Taking the lower limit as $\epsilon$ goes to zero,
for any equilibrium at demands $r_1=r_2=1+\epsilon$,
yields
\[
\liminf_{\epsilon\downarrow0}
\frac{\hat C(f)}{C_{opt}}
\ge
\lim_{\epsilon\downarrow0}
\frac{2(1+\epsilon)-1/4}
     {(1+\epsilon)+(1+\epsilon)^2/2}
=\frac76.
\]
\end{proof}

\paragraph{Network congestion games: Single-commodity.}
We now turn to single-commodity network congestion games and show that,
unlike the multi-commodity case, the special topology of these games enables a
positive result. In particular, given a prediction $\bar r$ and an error tolerance
parameter $\bar\eta$, 
we construct a $1$-consistent mechanism that admits
an equilibrium whose approximation ratio varies continuously with the input rate in
$[\max\{0,\bar r-\bar\eta\},\bar r+\bar\eta]$.
We assume $\bar r>0$ and $C_{opt}(\bar r)>0$.
For polynomial latencies with nonnegative coefficients, the latter
assumption implies $C_{opt}(r)>0$ for every $r>0$.

\begin{definition}[\textsc{ErrorTolerant($\bar\eta,\bar r$)}]
Let $0\le\bar\eta<\infty$, and let $\bar f^*$ be an optimal flow at
the predicted rate $\bar r$ whose positive edge support is acyclic\footnote{Such a flow is guaranteed to exist~\cite{phd/Roughgarden02}}.
Choose a maximum-cost $s$--$t$ path in its positive edge support, i.e.,
\(
p^*\in\arg\max_{p\in\mathcal P:\ \bar f_e^*>0\ \text{for all }e\in p}
\ell_p(\bar f^*).
\)
Let $u_e=\bar f_e^*+\bar\eta\boldsymbol{1}_{\{e\in p^*\}}$ and choose
$W\ge\sum_{e\in p^*}\ell_e(u_e)$.
Define
\[
\hat\ell_e(x)=
\begin{cases}
\ell_e(x),&0\le x\le u_e,\\
\max\{W,\ell_e(x)\},&x>u_e.
\end{cases}
\]
\end{definition}

In the next theorem we show that the \textsc{ErrorTolerant} mechanism achieves the combination of 1-consistency and error tolerance.

\begin{theorem}
\label{thm:SmoothnessSC}
In single-commodity network non-atomic congestion games, the
\textsc{ErrorTolerant($\bar\eta,\bar r$)} mechanism is $1$-consistent
and error tolerant for any input rate
$r\in[\max\{0,\bar r-\bar\eta\},\bar r+\bar\eta]$.
\end{theorem}

\begin{proof}
    Given a predicted input rate $\bar r$, let $\bar f^*$
    be an acyclic optimal flow under $\bar r$, and let
    $p^*$ be defined as above. Consider any input rate
    $r$, with $\bar r\leq r\leq\bar r+\bar\eta$,
    and let $f$ be the flow that routes $\bar r$ as
    $\bar f^*$ does, and the extra $r-\bar r$ via $p^*$,
    i.e., $f_e=\bar f_e^*$ for any $e\notin p^*$,
    and $f_e=\bar f_e^*+(r-\bar r)$ for any
    $e\in p^*$. Since $f_e\leq u_e$ for every
    $e\in E$, the modified latencies coincide with
    the original latencies at $f$.
    We first show that $f$ is an equilibrium for
    every $r\in[\bar r,\bar r+\bar\eta]$.
    For this, let $p_1$ be any path used by $f$,
    i.e., $f_{p_1}>0$, and let $p_2$ be any
    other path in the network. We establish
    the following claim:

    \begin{claim}
    \label{cl:smoothSingleNet}
    For every path $p_1$ used by $f$ and every
    alternative path $p_2$, it holds that
    \[
    \sum_{e\in p_1\setminus p_2}\hat\ell_e(f_e)
    \le
    \liminf_{\epsilon\downarrow0}
    \sum_{e\in p_2\setminus p_1}
    \hat\ell_e(f_e+\epsilon).
    \]
    \end{claim}
    
    \begin{proof}
        First, observe that $p^*$ remains a maximum-cost
        path in the positive edge support of $\bar f^*$
        under $f$. Indeed, the additional flow is routed
        exclusively through $p^*$, the original latencies
        are nondecreasing, and
        $\hat\ell_e(f_e)=\ell_e(f_e)$ for every $e\in E$.
        We distinguish two cases based on whether there
        exists an edge in $p_2\setminus p_1$ that does
        not belong to $p^*$.
        
        If $(p_2\setminus p_1)\setminus p^*
        \neq\emptyset$, then there exists an edge
        $e'\in p_2\setminus p_1$ such that
        $e'\notin p^*$. Since
        $f_{e'}=\bar f_{e'}^*=u_{e'}$,
        by the definition of the mechanism, for
        any $\epsilon>0$,
        \(
        \sum_{e\in p_2\setminus p_1}
        \hat\ell_e(f_e+\epsilon)
        \geq\hat\ell_{e'}(f_{e'}+\epsilon)
        \geq W.
        \)
        Consequently,
        \(
        \liminf_{\epsilon\downarrow0}
        \sum_{e\in p_2\setminus p_1}
        \hat\ell_e(f_e+\epsilon)\geq W.
        \)
        On the other hand, since $p_1$ belongs to
        the positive edge support of $\bar f^*$,
        \[
        \begin{aligned}
        \sum_{e\in p_1\setminus p_2}\hat\ell_e(f_e)
        &\leq\hat\ell_{p_1}(f)\\
        &\leq\hat\ell_{p^*}(f)\\
        &\leq\hat\ell_{p^*}(u)\\
        &=\sum_{e\in p^*}\ell_e(\bar f_e^*+\bar\eta)
        \leq W.
        \end{aligned}
        \]
        This establishes the claim in the first case.
        
        If $(p_2\setminus p_1)\setminus p^*
        =\emptyset$, then every edge in
        $p_2\setminus p_1$ belongs to $p^*$.
        We consider the alternative path
        \(
        p'=\bigl(p^*\setminus(p_2\setminus p_1)\bigr)
        \cup(p_1\setminus p_2).
        \)
        This is a valid $s$--$t$ path in the positive
        edge support of $\bar f^*$. 
        Since $p^*$ remains a maximum-cost path in
        this support under $f$, we have
        \(
        \hat\ell_{p^*}(f)\geq\hat\ell_{p'}(f).
        \)
        Furthermore, by the construction of $p'$,
        \(
        \hat\ell_{p_1}(f)+\hat\ell_{p^*}(f)
        =\hat\ell_{p_2}(f)+\hat\ell_{p'}(f).
        \)
        Combining these inequalities and cancelling
        the common edges of $p_1$ and $p_2$ yields
        \(
        \hat\ell_{p_2\setminus p_1}(f)
        \geq\hat\ell_{p_1\setminus p_2}(f).
        \)
        Consequently,
        \[
        \begin{aligned}
        \liminf_{\epsilon\downarrow0}
        \sum_{e\in p_2\setminus p_1}
        \hat\ell_e(f_e+\epsilon)
        &\geq
        \sum_{e\in p_2\setminus p_1}
        \hat\ell_e(f_e) \geq
        \sum_{e\in p_1\setminus p_2}
        \hat\ell_e(f_e),
        \end{aligned}
        \]
        where the first inequality follows from
        the nondecreasing modified latencies.
    \end{proof}

    By Claim~\ref{cl:smoothSingleNet}, we next show that the condition for $f$ bein an equilibrium is satisfied:
    \[
    \begin{aligned}
    \hat\ell_{p_1}(f)
    &=\sum_{e\in p_1\cap p_2}\hat\ell_e(f_e)
    +\sum_{e\in p_1\setminus p_2}\hat\ell_e(f_e)\\
    &\leq
    \sum_{e\in p_1\cap p_2}\hat\ell_e(f_e)
    +\liminf_{\epsilon\downarrow0}
    \sum_{e\in p_2\setminus p_1}
    \hat\ell_e(f_e+\epsilon)\\
    &=
    \liminf_{\epsilon\downarrow0}
    \hat\ell_{p_2}
    (f+\epsilon\boldsymbol{1}_{p_2}
    -\epsilon\boldsymbol{1}_{p_1}).
    \end{aligned}
    \]
    Thus, $f$ is a user equilibrium of the modified
    game for every $r\in[\bar r,\bar r+\bar\eta]$.
    Moreover, since $f_e\leq u_e$ for every edge,
    the modified latencies coincide with the
    original latencies at $f$. Consequently,
    $\hat C(f)=C(f)$ varies continuously with $r$,
    as the additional flow is routed through $p^*$. Since the mechanism ensures that the cost on this path does not increase abruptly for small changes in the flow, the error tolerance condition is satisfied.
    In particular, for $r=\bar r$, we recover
    $\bar f^*$ with
    $\hat C(\bar f^*)=C_{\mathrm{opt}}(\bar r)$.

    For $r\leq\bar r$, consider a Beckmann–McGuire–Winsten (BMW) equilibrium $f$ in the network with edge
    capacities $u_e$, i.e., a minimizer of the
    original Beckmann potential subject to
    $f_e\leq u_e$ for every $e\in E$~\cite{mor/CorreaSM04}.
    Existence and the corresponding variational inequalities follow from
    \cite{mor/CorreaSM04}. It is easy to see that this also a user equilibrium under the modified latencies. Any deviation that exceeds a capacity incurs a modified latency of at least $W$ on some edge, and therefore, cannot be profitable. 
    Finally, all BMW equilibria at a fixed rate
    have the same social cost, as follows from
    \cite[Lemma 3.3]{mor/CorreaSM04} and monotonicity. Since the capacitated feasible sets
    vary continuously with the input rate,
    compactness and continuity of the Beckmann
    potential imply that this common social cost
    is continuous in $r$.
    At $r=\bar r$, the optimal flow $\bar f^*$
    is such a minimizer, since it minimizes the
    modified potential and the two potentials
    coincide below the thresholds.
    Consequently, the equilibrium cost converges
    to $C_{\mathrm{opt}}(\bar r)$ as
    $r\uparrow\bar r$.
\end{proof}

\bibliographystyle{alpha}
\bibliography{sample}

@article{CaragiannisF21,
	Author = {Ioannis Caragiannis and Angelo Fanelli},
	Journal = {J. Comput. Syst. Sci.},
	Pages = {40--48},
	Title = {On approximate pure Nash equilibria in weighted congestion games with polynomial latencies},
	Volume = {117},
	Year = {2021}}

@article{Bilo18,
	Author = {Vittorio Bil{\`{o}}},
	Journal = {Theory Comput. Syst.},
	Number = {5},
	Pages = {1288--1317},
	Title = {A Unifying Tool for Bounding the Quality of Non-Cooperative Solutions in Weighted Congestion Games},
	Volume = {62},
	Year = {2018}}

@article{RT04,
	Author = {Tim Roughgarden and {\'{E}}va Tardos},
	Doi = {10.1016/j.geb.2003.06.004},
	Journal = {Games and Economic Behavior},
	Month = {may},
	Number = {2},
	Pages = {389--403},
	Publisher = {Elsevier {BV}},
	Title = {Bounding the inefficiency of equilibria in nonatomic congestion games},
	Volume = {47},
	Year = {2004}}

@article{Roughgarden15,
	Author = {Tim Roughgarden},
	Journal = {J. {ACM}},
	Number = {5},
	Pages = {32:1--32:42},
	Title = {Intrinsic Robustness of the Price of Anarchy},
	Volume = {62},
	Year = {2015}}

@article{algorithmica/ChristodoulouKS11,
	Author = {George Christodoulou and Elias Koutsoupias and Paul G. Spirakis},
	Journal = {Algorithmica},
	Number = {1},
	Pages = {116--140},
	Title = {On the Performance of Approximate Equilibria in Congestion Games},
	Volume = {61},
	Year = {2011}}

@article{Kollias13,
	Author = {Konstantinos Kollias},
	Journal = {Theory Comput. Syst.},
	Number = {3},
	Pages = {424--440},
	Title = {Nonpreemptive Coordination Mechanisms for Identical Machines},
	Volume = {53},
	Year = {2013}}

@article{arxiv:CKS08,
	Author = {George Christodoulou and Elias Koutsoupias and Paul G. Spirakis},
	Journal = {CoRR},
	Title = {On the performance of approximate equilibria in congestion games},
	Volume = {abs/0804.3160},
	Year = {2008}}

@inproceedings{BhattacharyaIKM14,
	Author = {Sayan Bhattacharya and Sungjin Im and Janardhan Kulkarni and Kamesh Munagala},
	Booktitle = {{ITCS}},
	Pages = {121--134},
	Publisher = {{ACM}},
	Title = {Coordination mechanisms from (almost) all scheduling policies},
	Year = {2014}}

@article{ColeCGMO15,
	Author = {Richard Cole and Jos{\'{e}} R. Correa and Vasilis Gkatzelis and Vahab S. Mirrokni and Neil Olver},
	Journal = {Games Econ. Behav.},
	Pages = {306--326},
	Title = {Decentralized utilitarian mechanisms for scheduling games},
	Volume = {92},
	Year = {2015}}

@inproceedings{esa/AbedH12,
  author       = {Fidaa Abed and
                  Chien{-}Chung Huang},
  title        = {Preemptive Coordination Mechanisms for Unrelated Machines},
  booktitle    = {{ESA}},
  series       = {Lecture Notes in Computer Science},
  volume       = {7501},
  pages        = {12--23},
  publisher    = {Springer},
  year         = {2012}
}

@article{mor/CorreaSM04,
	Author = {Jos{\'{e}} R. Correa and Andreas S. Schulz and Nicol{\'{a}}s E. Stier Moses},
	Journal = {Math. Oper. Res.},
	Number = {4},
	Pages = {961--976},
	Title = {Selfish Routing in Capacitated Networks},
	Volume = {29},
	Year = {2004}}

@inproceedings{approx/HansknechtKS14,
	Author = {Christoph Hansknecht and Max Klimm and Alexander Skopalik},
	Booktitle = {{APPROX-RANDOM}},
	Pages = {242--257},
	Publisher = {Schloss Dagstuhl - Leibniz-Zentrum f{\"{u}}r Informatik},
	Series = {LIPIcs},
	Title = {Approximate Pure Nash Equilibria in Weighted Congestion Games},
	Volume = {28},
	Year = {2014}}

@inproceedings{icalp/BhattacharyaKM14,
	Author = {Sayan Bhattacharya and Janardhan Kulkarni and Vahab S. Mirrokni},
	Booktitle = {{ICALP} {(1)}},
	Pages = {186--197},
	Publisher = {Springer},
	Series = {Lecture Notes in Computer Science},
	Title = {Coordination Mechanisms for Selfish Routing over Time on a Tree},
	Volume = {8572},
	Year = {2014}}

@phdthesis{phd/Roughgarden02,
	Author = {Tim Roughgarden},
	School = {Cornell University, {USA}},
	Title = {Selfish Routing},
	Year = {2002}}

@article{talg/CKK10,
	Author = {Ioannis Caragiannis and Christos Kaklamanis and Panagiotis Kanellopoulos},
	Journal = {{ACM} Trans. Algorithms},
	Number = {1},
	Pages = {13:1--13:31},
	Title = {Taxes for linear atomic congestion games},
	Volume = {7},
	Year = {2010}}

@article{tocs/GiannakopoulosP23,
	Author = {Yiannis Giannakopoulos and Diogo Po{\c{c}}as},
	Journal = {Theory Comput. Syst.},
	Number = {4},
	Pages = {855--876},
	Title = {A Unifying Approximate Potential for Weighted Congestion Games},
	Volume = {67},
	Year = {2023}}

@article{tocs/Fotakis10a,
	Author = {Dimitris Fotakis},
	Journal = {Theory Comput. Syst.},
	Number = {1},
	Pages = {218--249},
	Title = {Stackelberg Strategies for Atomic Congestion Games},
	Volume = {47},
	Year = {2010}}

@inproceedings{ITCS/BalkanskiGT23,
	Author = {Eric Balkanski and Vasilis Gkatzelis and Xizhi Tan},
	Booktitle = {{ITCS}},
	Pages = {11:1--11:22},
	Publisher = {Schloss Dagstuhl - Leibniz-Zentrum f{\"{u}}r Informatik},
	Series = {LIPIcs},
	Title = {Strategyproof Scheduling with Predictions},
	Volume = {251},
	Year = {2023}}

@misc{LM:alpsRepo,
	Author = {Lindermayr, Alexander and Megow, Nicole},
	Title = {Repository of works on algorithms with predictions},
	Year = {2023},
    howpublished = {\url{https://algorithms-with-predictions.github.io/}}
    }

@article{sigecom/BalkanskiGT23Overview,
	Author = {Eric Balkanski and Vasilis Gkatzelis and Xizhi Tan},
	Journal = {SIGecom Exch.},
	Number = {1},
	Pages = {54--57},
	Title = {Mechanism Design with Predictions: An Annotated Reading List},
	Volume = {21},
	Year = {2023}}

@article{cacm/MitzenmacherV22,
	Author = {Michael Mitzenmacher and Sergei Vassilvitskii},
	Journal = {Commun. {ACM}},
	Number = {7},
	Pages = {33--35},
	Title = {Algorithms with predictions},
	Volume = {65},
	Year = {2022}}

@inproceedings{CBKST24,
	Author = {Riccardo Colini{-}Baldeschi and Sophie Klumper and Guido Sch{\"{a}}fer and Artem Tsikiridis},
	Booktitle = {{EC}},
	Pages = {1134--1154},
	Publisher = {{ACM}},
	Title = {To Trust or Not to Trust: Assignment Mechanisms with Predictions in the Private Graph Model},
	Year = {2024}}

@article{mor/AgrawalBGOT24,
	Author = {Priyank Agrawal and Eric Balkanski and Vasilis Gkatzelis and Tingting Ou and Xizhi Tan},
	Journal = {Math. Oper. Res.},
	Number = {4},
	Pages = {2626--2651},
	Title = {Learning-Augmented Mechanism Design: Leveraging Predictions for Facility Location},
	Volume = {49},
	Year = {2024}}

@article{teac/BV19,
	Author = {Vittorio Bil{\`{o}} and Cosimo Vinci},
	Journal = {{ACM} Trans. Economics and Comput.},
	Number = {3},
	Pages = {15:1--15:36},
	Title = {Dynamic Taxes for Polynomial Congestion Games},
	Volume = {7},
	Year = {2019}}

@article{ior/PG24,
	Author = {Dario Paccagnan and Martin Gairing},
	Journal = {Oper. Res.},
	Number = {3},
	Pages = {966--982},
	Title = {In Congestion Games, Taxes Achieve Optimal Approximation},
	Volume = {72},
	Year = {2024}}

@article{AzarFJMS15,
	Author = {Yossi Azar and Lisa Fleischer and Kamal Jain and Vahab S. Mirrokni and Zoya Svitkina},
	Journal = {Oper. Res.},
	Number = {3},
	Pages = {489--500},
	Title = {Optimal Coordination Mechanisms for Unrelated Machine Scheduling},
	Volume = {63},
	Year = {2015}}

@article{mor/BonifaciHS10,
	Author = {Vincenzo Bonifaci and Tobias Harks and Guido Sch{\"{a}}fer},
	Journal = {Math. Oper. Res.},
	Number = {2},
	Pages = {330--346},
	Title = {Stackelberg Routing in Arbitrary Networks},
	Volume = {35},
	Year = {2010}}

@article{algorithmica/Caragiannis13,
	Author = {Ioannis Caragiannis},
	Journal = {Algorithmica},
	Number = {3},
	Pages = {512--540},
	Title = {Efficient Coordination Mechanisms for Unrelated Machine Scheduling},
	Volume = {66},
	Year = {2013}}

@article{tocs/CaragiannisF19,
	Author = {Ioannis Caragiannis and Angelo Fanelli},
	Journal = {Theory Comput. Syst.},
	Number = {1},
	Pages = {114--127},
	Title = {An Almost Ideal Coordination Mechanism for Unrelated Machine Scheduling},
	Volume = {63},
	Year = {2019}}

@article{siamcomp/Roughgarden04,
	Author = {Tim Roughgarden},
	Journal = {{SIAM} J. Comput.},
	Number = {2},
	Pages = {332--350},
	Title = {Stackelberg Scheduling Strategies},
	Volume = {33},
	Year = {2004}}

@article{or/GkatzelisKR16,
	Author = {Vasilis Gkatzelis and Konstantinos Kollias and Tim Roughgarden},
	Journal = {Oper. Res.},
	Number = {6},
	Pages = {1230--1238},
	Title = {Optimal Cost-Sharing in General Resource Selection Games},
	Volume = {64},
	Year = {2016}}

@inproceedings{focs/KarakostasK04,
	Author = {George Karakostas and Stavros G. Kolliopoulos},
	Booktitle = {{FOCS}},
	Pages = {268--276},
	Publisher = {{IEEE} Computer Society},
	Title = {Edge Pricing of Multicommodity Networks for Heterogeneous Selfish Users},
	Year = {2004}}

@article{siamcomp/ChristodoulouGG19,
	Author = {George Christodoulou and Martin Gairing and Yiannis Giannakopoulos and Paul G. Spirakis},
	Journal = {{SIAM} J. Comput.},
	Number = {5},
	Pages = {1544--1582},
	Title = {The Price of Stability of Weighted Congestion Games},
	Volume = {48},
	Year = {2019}}

@article{teac/CG16:PoS,
	Author = {George Christodoulou and Martin Gairing},
	Journal = {{ACM} Trans. Economics and Comput.},
	Number = {2},
	Pages = {10:1--10:17},
	Title = {Price of Stability in Polynomial Congestion Games},
	Volume = {4},
	Year = {2016}}

@article{teac/BhawalkarGR14:PoAWeighted,
  author       = {Kshipra Bhawalkar and
                  Martin Gairing and
                  Tim Roughgarden},
  title        = {Weighted Congestion Games: The Price of Anarchy, Universal Worst-Case
                  Examples, and Tightness},
  journal      = {{ACM} Trans. Economics and Comput.},
  volume       = {2},
  number       = {4},
  pages        = {14:1--14:23},
  year         = {2014}
}

@article{siamcomp/AlandDGMS11,
	Author = {Sebastian Aland and Dominic Dumrauf and Martin Gairing and Burkhard Monien and Florian Schoppmann},
	Journal = {{SIAM} J. Comput.},
	Number = {5},
	Pages = {1211--1233},
	Title = {Exact Price of Anarchy for Polynomial Congestion Games},
	Volume = {40},
	Year = {2011}}

@inproceedings{focs/FleischerJM04,
	Author = {Lisa Fleischer and Kamal Jain and Mohammad Mahdian},
	Booktitle = {{FOCS}},
	Pages = {277--285},
	Publisher = {{IEEE} Computer Society},
	Title = {Tolls for Heterogeneous Selfish Users in Multicommodity Networks and Generalized Congestion Games},
	Year = {2004}}

@inproceedings{stoc/ColeDR03,
	Author = {Richard Cole and Yevgeniy Dodis and Tim Roughgarden},
	Booktitle = {{STOC}},
	Pages = {521--530},
	Publisher = {{ACM}},
	Title = {Pricing network edges for heterogeneous selfish users},
	Year = {2003}}

@article{jcss/ColeDR06Taxes,
	Author = {Richard Cole and Yevgeniy Dodis and Tim Roughgarden},
	Journal = {J. Comput. Syst. Sci.},
	Number = {3},
	Pages = {444--467},
	Title = {How much can taxes help selfish routing?},
	Volume = {72},
	Year = {2006}}

@inproceedings{sigecom/BhaskarCKS16,
	Author = {Umang Bhaskar and Yu Cheng and Young Kun{-}Ko and Chaitanya Swamy},
	Booktitle = {{EC}},
	Pages = {479--496},
	Publisher = {{ACM}},
	Title = {Hardness Results for Signaling in Bayesian Zero-Sum and Network Routing Games},
	Year = {2016}}

@article{tcs/ImmorlicaLMS09,
	Author = {Nicole Immorlica and Li (Erran) Li and Vahab S. Mirrokni and Andreas S. Schulz},
	Journal = {Theor. Comput. Sci.},
	Number = {17},
	Pages = {1589--1598},
	Title = {Coordination mechanisms for selfish scheduling},
	Volume = {410},
	Year = {2009}}

@inproceedings{isaac/WangYC16:ImprovedTax,
	Author = {Te{-}Li Wang and Chih{-}Kuan Yeh and Ho{-}Lin Chen},
	Booktitle = {{ISAAC}},
	Pages = {61:1--61:12},
	Publisher = {Schloss Dagstuhl - Leibniz-Zentrum f{\"{u}}r Informatik},
	Series = {LIPIcs},
	Title = {An Improved Tax Scheme for Selfish Routing},
	Volume = {64},
	Year = {2016}}

@book{Beckmann1956-,
  title={Studies in the Economics of Transportation},
  author={Beckmann, M.J. and McGuire, C.B. and Winsten, C.B.},
  year={1956},
  publisher={Cowles Commission for Research in Economics}
}

@inproceedings{KoutsoupiasP99,
	Author = {Elias Koutsoupias and Christos H. Papadimitriou},
	Booktitle = {{STACS}},
	Pages = {404--413},
	Publisher = {Springer},
	Series = {Lecture Notes in Computer Science},
	Title = {Worst-case Equilibria},
	Volume = {1563},
	Year = {1999}}

@article{RoughgardenT02,
	Author = {Tim Roughgarden and {\'{E}}va Tardos},
	Journal = {J. {ACM}},
	Number = {2},
	Pages = {236--259},
	Title = {How bad is selfish routing?},
	Volume = {49},
	Year = {2002}}

@inproceedings{stoc/ChristodoulouK05,
	Author = {George Christodoulou and Elias Koutsoupias},
	Booktitle = {{STOC}},
	Pages = {67--73},
	Publisher = {{ACM}},
	Title = {The price of anarchy of finite congestion games},
	Year = {2005}}

@article{Rosenthal1973,
	Author = {Rosenthal, Robert W},
	Journal = {International Journal of Game Theory},
	Month = dec,
	Number = 1,
	Pages = {65--67},
	Title = {A class of games possessing pure-strategy Nash equilibria},
	Volume = 2,
	Year = 1973}

@article{MondererShapley:Potential,
	Author = {Dov Monderer and Lloyd S. Shapley},
	Doi = {https://doi.org/10.1006/game.1996.0044},
	Issn = {0899-8256},
	Journal = {Games and Economic Behavior},
	Number = {1},
	Pages = {124-143},
	Title = {Potential Games},
	Volume = {14},
	Year = {1996}}

@inproceedings{CCSV25:predictions,
  author       = {George Christodoulou and
                  Vasilis Christoforidis and
                  Alkmini Sgouritsa and
                  Ioannis Vlachos},
  title        = {Improving the Price of Anarchy via Predictions in Parallel-Link Networks},
  booktitle    = {{WWW}},
  pages        = {309--317},
  publisher    = {{ACM}},
  year         = {2026}
}

@inproceedings{stoc/AwerbuchAE05,
	Author = {Baruch Awerbuch and Yossi Azar and Amir Epstein},
	Booktitle = {{STOC}},
	Pages = {57--66},
	Publisher = {{ACM}},
	Title = {The Price of Routing Unsplittable Flow},
	Year = {2005}}

@inproceedings{CaragiannisGV17,
	Author = {Ioannis Caragiannis and Vasilis Gkatzelis and Cosimo Vinci},
	Booktitle = {{WINE}},
	Pages = {74--87},
	Publisher = {Springer},
	Series = {Lecture Notes in Computer Science},
	Title = {Coordination Mechanisms, Cost-Sharing, and Approximation Algorithms for Scheduling},
	Volume = {10660},
	Year = {2017}}

@article{mst/CohenDT11,
  author       = {Johanne Cohen and
                  Christoph D{\"{u}}rr and
                  Kim Thang Nguyen},
  title        = {Non-clairvoyant Scheduling Games},
  journal      = {Theory Comput. Syst.},
  volume       = {49},
  number       = {1},
  pages        = {3--23},
  year         = {2011}
}

@inproceedings{esa/AbedCH14,
  author       = {Fidaa Abed and
                  Jos{\'{e}} R. Correa and
                  Chien{-}Chung Huang},
  title        = {Optimal Coordination Mechanisms for Multi-job Scheduling Games},
  booktitle    = {{ESA}},
  series       = {Lecture Notes in Computer Science},
  volume       = {8737},
  pages        = {13--24},
  publisher    = {Springer},
  year         = {2014}
}

@article{tcs/ChristodoulouKN09:CMs,
	Author = {George Christodoulou and Elias Koutsoupias and Akash Nanavati},
	Journal = {Theor. Comput. Sci.},
	Number = {36},
	Pages = {3327--3336},
	Title = {Coordination mechanisms},
	Volume = {410},
	Year = {2009}}

@article{tcs/FotakisKS05,
	Author = {Dimitris Fotakis and Spyros C. Kontogiannis and Paul G. Spirakis},
	Journal = {Theor. Comput. Sci.},
	Number = {2-3},
	Pages = {226--239},
	Title = {Selfish unsplittable flows},
	Volume = {348},
	Year = {2005}}

@inproceedings{icalp/Colini-Baldeschi18,
	Author = {Riccardo Colini{-}Baldeschi and Max Klimm and Marco Scarsini},
	Booktitle = {{ICALP}},
	Pages = {151:1--151:14},
	Publisher = {Schloss Dagstuhl - Leibniz-Zentrum f{\"{u}}r Informatik},
	Series = {LIPIcs},
	Title = {Demand-Independent Optimal Tolls},
	Volume = {107},
	Year = {2018}}

@article{mor/GiannakopoulosNS22,
  author       = {Yiannis Giannakopoulos and
                  Georgy Noarov and
                  Andreas S. Schulz},
  title        = {Computing Approximate Equilibria in Weighted Congestion Games via
                  Best-Responses},
  journal      = {Math. Oper. Res.},
  volume       = {47},
  number       = {1},
  pages        = {643--664},
  year         = {2022}
}

@inproceedings{sigecom/LuW024,
	Author = {Pinyan Lu and Zongqi Wan and Jialin Zhang},
	Booktitle = {{EC}},
	Pages = {1155--1183},
	Publisher = {{ACM}},
	Title = {Competitive Auctions with Imperfect Predictions},
	Year = {2024}}

@article{siamcomp/AnshelevichDKTWR08:PoS,
	Author = {Elliot Anshelevich and Anirban Dasgupta and Jon M. Kleinberg and {\'{E}}va Tardos and Tom Wexler and Tim Roughgarden},
	Journal = {{SIAM} J. Comput.},
	Number = {4},
	Pages = {1602--1623},
	Title = {The Price of Stability for Network Design with Fair Cost Allocation},
	Volume = {38},
	Year = {2008}}

@article{ior/CGS24,
	Author = {George Christodoulou and Vasilis Gkatzelis and Alkmini Sgouritsa},
	Journal = {Oper. Res.},
	Number = {1},
	Pages = {167--184},
	Title = {Resource-Aware Cost-Sharing Methods for Scheduling Games},
	Volume = {72},
	Year = {2024}}

@article{algorithmica/CFKKM11,
	Author = {Ioannis Caragiannis and Michele Flammini and Christos Kaklamanis and Panagiotis Kanellopoulos and Luca Moscardelli},
	Journal = {Algorithmica},
	Number = {3},
	Pages = {606--637},
	Title = {Tight Bounds for Selfish and Greedy Load Balancing},
	Volume = {61},
	Year = {2011}}

@inproceedings{esa/ChristodoulouK05,
	Author = {George Christodoulou and Elias Koutsoupias},
	Booktitle = {{ESA}},
	Pages = {59--70},
	Publisher = {Springer},
	Series = {Lecture Notes in Computer Science},
	Title = {On the Price of Anarchy and Stability of Correlated Equilibria of Linear Congestion Games},
	Volume = {3669},
	Year = {2005}}

@article{algorithmica/ChristodoulouMP14,
	Author = {Giorgos Christodoulou and Kurt Mehlhorn and Evangelia Pyrga},
	Journal = {Algorithmica},
	Number = {3},
	Pages = {619--640},
	Title = {Improving the Price of Anarchy for Selfish Routing via Coordination Mechanisms},
	Volume = {69},
	Year = {2014}}

@article{transci/BernsteinS94:UserEq,
	Author = {David Bernstein and Tony E. Smith},
	Journal = {Transp. Sci.},
	Number = {3},
	Pages = {221--235},
	Title = {Equilibria for Networks with Lower Semicontinuous Costs: With an Application to Congestion Pricing},
	Volume = {28},
	Year = {1994}}

@article{FotakisSpirakis08,
	Author = {Dimitris Fotakis and Paul G. Spirakis},
	Journal = {Internet Math.},
	Number = {4},
	Pages = {343--363},
	Title = {Cost-Balancing Tolls for Atomic Network Congestion Games},
	Volume = {5},
	Year = {2008}}

@inproceedings{soda/Gkatzelis0T25,
	Author = {Vasilis Gkatzelis and Daniel Schoepflin and Xizhi Tan},
	Booktitle = {{SODA}},
	Pages = {2629--2655},
	Publisher = {{SIAM}},
	Title = {Clock Auctions Augmented with Unreliable Advice},
	Year = {2025}}

@book{NRTV2007a,
	Editor = {Noam Nisan and Tim Roughgarden and {\'{E}}va Tardos and Vijay V. Vazirani},
	Publisher = {Cambridge University Press},
	Title = {Algorithmic Game Theory},
	Year = {2007}}

@inproceedings{CSV24,
	Author = {George Christodoulou and Alkmini Sgouritsa and Ioannis Vlachos},
	Booktitle = {NeurIPS},
	Title = {Mechanism design augmented with output advice},
	Year = {2024}}

@inproceedings{sigecom/GkatzelisKST22,
	Author = {Vasilis Gkatzelis and Kostas Kollias and Alkmini Sgouritsa and Xizhi Tan},
	Booktitle = {{EC}},
	Pages = {529--557},
	Publisher = {{ACM}},
	Title = {Improved Price of Anarchy via Predictions},
	Year = {2022}}

@inproceedings{soda/Kashaev26,
  author    = {Danish Kashaev},
  title     = {Selfish, Local and Online Scheduling via Vector Fitting},
  booktitle = {Proceedings of the 2026 Annual {ACM-SIAM} Symposium on Discrete Algorithms ({SODA})},
  pages     = {3841--3871},
  year      = {2026}
}

@inproceedings{stoc/BlumHLR08,
  author       = {Avrim Blum and
                  MohammadTaghi Hajiaghayi and
                  Katrina Ligett and
                  Aaron Roth},
  title        = {Regret minimization and the price of total anarchy},
  booktitle    = {{STOC}},
  pages        = {373--382},
  publisher    = {{ACM}},
  year         = {2008}
}

@article{LvDKLW15:trafficPrediction,
  author       = {Yisheng Lv and
                  Yanjie Duan and
                  Wenwen Kang and
                  Zhengxi Li and
                  Fei{-}Yue Wang},
  title        = {Traffic Flow Prediction With Big Data: {A} Deep Learning Approach},
  journal      = {{IEEE} Trans. Intell. Transp. Syst.},
  volume       = {16},
  number       = {2},
  pages        = {865--873},
  year         = {2015}
}

@inproceedings{nips/Elenter0DL24:britleness,
  author       = {Alex Elenter and
                  Spyros Angelopoulos and
                  Christoph D{\"{u}}rr and
                  Yanni Lefki},
  title        = {Overcoming Brittleness in Pareto-Optimal Learning Augmented Algorithms},
  booktitle    = {NeurIPS},
  year         = {2024}
}

@inproceedings{aistats/BenomarP25,
  author       = {Ziyad Benomar and
                  Vianney Perchet},
  title        = {On Tradeoffs in Learning-Augmented Algorithms},
  booktitle    = {{AISTATS}},
  series       = {Proceedings of Machine Learning Research},
  volume       = {258},
  pages        = {802--810},
  publisher    = {{PMLR}},
  year         = {2025}
}

\appendix

\clearpage
\section{Arithmetic Lemmas}
\label{appendix:Arithmetic}

\begin{restatable}{lemma}{ArithmeticNonAtomic}\label{lem:arithmeticNonAtomic}
For any integer $k \in [1,d]$, non-negative real numbers $x,y$ and real number $\lambda > 0,$ it holds that

$$x^ky \leq \max\left\{\frac{1}{4\lambda},\frac{d}{(d+1)\sqrt[d]{\lambda(d+1)}}\right\}x^{k+1} + \lambda y^{k+1}.$$
\end{restatable}

\begin{proof}
Consider first the following arithmetic inequality. For any integer $k \in [1,d]$, non-negative real numbers $x,y$ and real number $\lambda \geq 1$, it holds that: 

$$x^ky \leq \frac{k}{(k+1)\sqrt[k]{\lambda(k+1)}}x^{k+1} + \lambda y^{k+1}.$$ 

This is used in \cite{algorithmica/ChristodoulouKS11}, but here we parameterize it over a parameter $\lambda$ which will be determined later in order to optimize the robustness bound. Considering only positive $k>0$, the coefficient $\frac{k}{(k+1)\sqrt[k]{\lambda(k+1)}}$ is zero for $k=0$ and strictly increasing when $\lambda \geq 1$. In the case of $\lambda < 1$ it achieves its minimum at some value $k_{min} > 0$ and is strictly decreasing for $k < k_{min}$ and strictly increasing for $k > k_{min}$. Overall, since $k \in [1,d]$, the expression achieves its maximum either for $k=1$ or for $k=d$, and therefore:

    $$x^ky \leq \max\left\{\frac{1}{4\lambda},\frac{d}{(d+1)\sqrt[d]{\lambda(d+1)}}\right\}x^{k+1} + \lambda y^{k+1}.$$
\end{proof}

\begin{lemma}\label{lem:polyAlgebra}
    For any integers $x\geq 1$, the following hold:
    \begin{enumerate}[{(1)}]
        \item $\sum_{i=1}^x \left(i^k-(i-1)^k\right)= x^k$,   for any integer $k\ge1\,;$
        \item $(k+1)x^k\geq x^{k+1}-(x-1)^{k+1}$,  for any integer $k\ge0\,;$
        \item $(k+1)\sum_{i=1}^x i^k\geq x^{k+1}$, for any integer $k\ge0\,;$
    \end{enumerate}
\end{lemma}

\begin{proof}
    The first equality holds for $k\ge1$ because
    $$\sum_{i=1}^x \left(i^k-(i-1)^k\right) = x^k+\sum_{i=1}^{x-1} i^k-\sum_{i=1}^{x} (i-1)^k= x^k+\sum_{i=1}^{x-1} i^k-\sum_{i=0}^{x-1} i^k = x^k\,. $$
    
    For the second inequality we will show that  
    $f(x,k)=(k+1)x^k - x^{k+1}+(x-1)^{k+1}\geq 0$ by induction on $k$. For $k=0$, $f(x,0)=0$ for any $x\geq 1$. Consider any $k\geq 1$, and suppose for the inductive hypothesis that $f(x,k-1)\geq 0$ for all $x\geq 1$. Taking the derivative w.r.t. $x$, $df(x,k)/dx = (k+1)f(x,k-1)\geq 0$, for all $x\geq 1$, meaning that $f(x,k)$ is non-decreasing w.r.t. $x$. Since, $f(1,k)=k\geq 0$, $f(x,k)\geq 0$, for any $x\geq 1$ and the statement follows.

The third inequality is derived immediately by the second one as follows: $$(k+1)\sum_{i=1}^x i^k=\sum_{i=1}^x (k+1)i^k\geq \sum_{i=1}^x \left(i^{k+1}-(i-1)^{k+1}\right)=x^{k+1}\,.$$ 
\end{proof}

\begin{lemma}\label{lem:ineqAtomic}
    For every $x,y,z \in \mathbb{N}$ and real number $\beta\ge1$:
    \begin{align*}
        y\left(x+\frac{1}{\beta}\right) \le \frac{1}{2z+1}x^2 + \frac{z^2+(1+\frac{2}{\beta})z+\frac{\beta+3}{4\beta}}{2z+1}y^2
    \end{align*}
\end{lemma}

\begin{proof}
    We subtract the left hand side of the statement from the right hand side, multiply by $2z+1$, and use $f(x,y,\beta)$ to denote the result. Then, we have:

    \begin{align*}
        f(x,y,\beta)& = x^2 + \left(z^2+\left(1+\frac{2}{\beta}\right)z+\frac{\beta+3}{4\beta}\right)y^2-(2z+1)y\left(x+\frac{1}{\beta}\right) \\
        & = 
        \left(x-\frac{2z+1}{2}y\right)^2+\left(\frac{2}{\beta}z+\frac{3}{4\beta}\right)y^2-\frac{2z+1} {\beta}y \\
        & = \left(x-\frac{2z+1}{2}y\right)^2+\frac{(8z+3)y-8z-4}{4}\frac{y}{\beta}\,. \\
    \end{align*}
    We want to show that $f(x,y,\beta) \ge 0$ for all nonnegative integers $x,y,z$. For $y=0$, this is clearly true. For any $y \ge 2$, $f(x,y,\beta)$ is positive since $2(8z+3)>8z+4$. It remains to check the case $y=1$; we have that $f(x,1,\beta) = x^2-(2z+1)x+z^2+z+\frac{1}{4}-\frac{1}{4\beta} = (x-z)(x-z-1) +\frac{1}{4}-\frac{1}{4\beta} \geq 0\,,$ which is nonnegative for $x,z$ integers and $\beta \geq 1$. So, the lemma follows.
\end{proof}

\end{document}